\documentclass[12pt, letterpaper]{article}

\usepackage[utf8]{inputenc}
\usepackage[driver=dvipdfm,margin=1in]{geometry}
\usepackage{mathtools}
\usepackage{amsmath}
\usepackage{amsthm}
\usepackage{amssymb}
\usepackage{setspace}
\usepackage{booktabs}
\usepackage{tabularx}
\usepackage{amsmath} 
\usepackage{amssymb}
\usepackage{bm}
\usepackage{ascmac}
\usepackage{setspace}
\usepackage[at]{easylist}
\usepackage{empheq}

\mathtoolsset{showonlyrefs,showmanualtags}  

\makeatletter

\usepackage{amsthm}
\usepackage{amsfonts}
\usepackage{bbm}
\usepackage{graphics}
\usepackage{enumerate}
\usepackage{float}
\usepackage{caption}
\usepackage{subcaption}
\usepackage{natbib}
\usepackage{epigraph}
\usepackage{epigraph}

\theoremstyle{definition}
\newtheorem{thm}{Theorem}
\newtheorem{lem}[thm]{Lemma}

\theoremstyle{remark}

\newtheorem{remark}{Remark}
\newtheorem{claim}{Claim}[subsection]

\newtheorem{case}{Case}

\usepackage[hyphens]{url}
\usepackage[colorlinks,urlcolor=blue, citecolor=blue, menucolor=blue]{hyperref}

\theoremstyle{plain}

\providecommand{\keywords}[1]
{
  \small	
  \textbf{Keywords:} #1
}
\providecommand{\JEL}[1]
{
  \small	
  \textbf{JEL Codes:} #1
}

\usepackage{algorithm,algpseudocode}

\usepackage{amsmath}
\DeclareMathOperator*{\argmax}{arg\,max}

\newcommand{\Real}{\mathbb{R}}

\newcommand{\Regret}{\mathrm{Reg}}
\newcommand{\Reg}{\mathrm{Reg}}

\newcommand{\regret}{\mathrm{reg}} 
\newcommand{\reg}{\mathrm{reg}}

\newcommand{\Ex}{\mathbb{E}}
\newcommand{\Prob}{\mathbb{P}}
\newcommand{\eps}{\epsilon}

\newcommand{\Normal}{\mathcal{N}}

\newcommand{\Ind}{\bm{1}}
\newcommand{\polylog}{\text{polylog}}
\newcommand{\tilO}{\tilde{O}}
\newcommand{\tilOmega}{\tilde{\Omega}}
\newcommand{\tilTheta}{\tilde{\Theta}}
\newcommand{\tiltheta}{\tilde{\theta}}

\newcommand{\sgn}{\mathrm{sgn}}

\newcommand{\erf}{\mathrm{erf}}

\newcommand{\EA}{\mathcal{A}}
\newcommand{\EB}{\mathcal{B}}
\newcommand{\EC}{\mathcal{C}}

\newcommand{\EE}{\mathcal{E}}

\newcommand{\EP}{\mathcal{P}}

\newcommand{\ET}{\mathcal{T}}
\newcommand{\EU}{\mathcal{U}}

\newcommand{\EX}{\mathcal{X}}

\newcommand{\EZ}{\mathcal{Z}}

\newcommand{\trunkN}{\mathcal{N}^{\mathrm{tr}}}

\newcommand{\Ehatzt}{Z_t}
\newcommand{\Ehatz}[1]{Z_{#1}}
\newcommand{\hatr}{\hat{r}}
\newcommand{\ist}{b^*_t}

\newcommand{\Unif}{\mathrm{Unif}}

\newcommand{\Creg}{C_{\mathrm{reg}}}

\newcommand{\Cezl}{C_{\mathrm{l}}} 
\newcommand{\Cezu}{C_{\mathrm{u}}}

\newcommand{\Cshrink}{C_{\mathrm{shrink}}}
\newcommand{\Ceps}{C_{\eps}}

\newcommand{\Ctemp}{C_2}
\newcommand{\Cupdate}{C_{\mathrm{update}}}
\newcommand{\ACC}{\mathrm{ACC}}
\newcommand{\Follow}{\ET_\mathrm{Obey}}
\newcommand{\Deviate}{\ET_\mathrm{Deviate}}
\newcommand{\Fence}{\ET_\mathrm{OtF}}
\newcommand{\ACCFollow}{\mathrm{ACC}_{\mathrm{obey}}}
\newcommand{\ACCDeviate}{\mathrm{ACC}_{\mathrm{deviate}}}
\newcommand{\ACCFence}{\mathrm{ACC}_{\mathrm{OtF}}}

\newcommand{\Cternary}{C_{\mathrm{ter}}}
\newcommand{\CfenceL}{C_{\mathrm{OtF}}^L}
\newcommand{\CfenceU}{C_{\mathrm{OtF}}^U}
\newcommand{\Cregt}{C_{\mathrm{regt}}}
\newcommand{\Cregttwo}{C_{\mathrm{regt}}} 
\newcommand{\Cgeomup}{C_w}
\newcommand{\Cgeomuptmp}{C_{w,2}}

\newcommand{\Runnum}{5000}
\newcommand{\Nsingle}{10,000}

\allowdisplaybreaks

\makeatother

\begin{document}

\title{Deviation-Based Learning: Training Recommender Systems Using Informed User Choice\thanks{We are grateful to Alex Bloedel, Jin-Wook Chang, Vitor Farinha Luz, Kohei Kawaguchi, Yichuan Lou, Daisuke Nakajima, Foster Provost, Wing Suen, and the seminar participants at the Happy Hour Seminar (online), the 27th Decentralization Conference (online), Otaru University of Commerce, the Information-Based Induction Sciences Workshop (IBISML) 2021 (online), the CUHK-HKU-HKUST Theory Seminar (online), the 2nd UTMD Conference (online), the Marketplace Innovations Workshop (online), and the North American Meeting (Miami), Australasia Meeting (online), and Asian Meeting (Tokyo) of the Econometric Society. All remaining errors are our own.}}

\author{
Junpei Komiyama\thanks{Leonard N. Stern School of Business, New York University, 44 West 4th Street, New York, NY 10012, United States. E-mail: \href{mailto:junpei.komiyama@gmail.com}{junpei.komiyama@gmail.com}.}
\and
Shunya Noda\thanks{Graduate School of Economics, University of Tokyo, 7-3-1 Hongo, Tokyo, 113-0033, Japan. E-mail: \href{mailto:shunya.noda@gmail.com}{shunya.noda@gmail.com}. Noda has been supported by the Social Sciences and Humanities Research Council of Canada and JSPS KAKENHI Grant Number JP22K13361.}
}

\date{First Draft: September 22, 2021; Last Updated: \today}

\maketitle

\begin{abstract}
This paper proposes a new approach to training recommender systems called \emph{deviation-based learning}. The recommender and rational users have different knowledge. The recommender learns user knowledge by observing what action users take upon receiving recommendations. Learning eventually stalls if the recommender always suggests a choice: Before the recommender completes learning, users start following the recommendations blindly, and their choices do not reflect their knowledge. The learning rate and social welfare improve substantially if the recommender abstains from recommending a particular choice when she predicts that multiple alternatives will produce a similar payoff.
\end{abstract}


\JEL{C44, D82, D83}

\vspace{1em}

\pagebreak

\section{Introduction}

In every day of our life, our choices rely on recommendations made by others based on their knowledge and experience.
The prosperity of online platforms and artificial intelligence has enabled us to develop data-based recommendations, and many systems have been implemented in practice.
Successful examples include e-commerce (Amazon), movies (Netflix), music (Spotify), restaurants (Yelp), sightseeing spots (TripAdvisor), hotels (Booking.com), classes (RateMyProfessors), hospitals (RateMD), and route directions by car navigation apps (Google Maps). 
These ``recommender systems''\footnote{In a narrow sense, a ``recommender system'' is defined as an algorithm for predicting rating users would enter. For example, \cite{AT2005_nextrecommend} state ``In its most common formulation, the recommendation problem is reduced to the problem of estimating ratings for the items that have not been seen by a user'' (p.~734). Our system is not a ``recommender system'' in this narrow sense because we do not utilize ratings. This paper adopts a broader definition of ``recommender system'' to denote any mechanism recommending arms (items or actions) to help users make better decisions.} are helping us to make better decisions.

The advantages of the data-based recommender systems can be classified into two groups. First, the system can leverage experiences of the most knowledgeable experts. Once the system learns experts' behavior using data, the system can report what a user would do if he had expert knowledge. Accordingly, with the help of the recommender system, all users can optimize their payoffs even when they have no experience with the problem they are facing. Second, the system can utilize information that an individual cannot access easily or quickly. For example, restaurant-reservation systems present the list of all available reservation slots at that moment, and online travel agencies provide the prices and available rooms of hotels. These conditions change over time; thus, it would be very difficult for an individual user to keep up to the minute with the latest conditions on their own. Accordingly, even experts benefit from the information provided by recommender systems.

One of the largest challenges in developing a recommender system is to predict users' payoffs associated with specific alternatives.
Real-world recommenders always confront the problem of insufficient initial experimentation (known as the ``cold start'' problem). Utilization of feedback provided by users is necessary, but such data are often incomplete and insufficient.
In particular, the system can rarely observe information about users' payoffs, which is crucial in many learning methods (e.g., reinforcement learning and algorithms to solve the multi-armed bandit problem).
As a proxy for payoffs, many recommender systems have adopted \emph{rating-based learning}, which substitutes the ratings submitted by the users for the true payoffs of users.
Nevertheless, a number of previous studies have reported that user-generated ratings often involve various types of biases and are not very informative signals of users' true payoffs \citep[e.g.,][]{salganik2006experimental,muchnik2013social,luca2016fake}.

In this paper, we propose a new approach to training recommender systems called \emph{deviation-based learning}. In our model, a recommender (she) faces many rational users (he) sequentially. Neither users' payoffs nor ratings are available.
Instead, we train a recommender system using data about past recommendations and users' actions taken after receiving recommendations. By focusing on the relationship between recommendation and action choice, the recommender can infer the user's knowledge. For example, if the recommender has not yet been well-trained, expert users often deviate from her recommendations. On the flip side of the coin, upon observing expert users' deviations, the recommender can recognize that she misestimated the underlying state. Conversely, if a user follows the recommendation while the recommender is not perfectly sure whether the user would follow it, then the recommender can improve her confidence in the accuracy of her recommendations. We refer to this approach as ``deviation-based learning'' because these two examples, both based on deviations, represent the most primitive ways of extracting users' knowledge from choices given information.

We evaluate the tradeoff between choice efficiency and communication complexity in deviation-based learning. If the recommender could send a more informative message to users, then users can better understand the recommender's information and make a better choice. However, because simpler communications are preferred in practice, the real-world recommender system often attempts to make the recommendation as simple as possible. By analyzing a stylized environment, we demonstrate that almost all gain is obtained by slightly enriching the communication from a simple straightforward recommendation, i.e., just to inform the estimated-to-be-better choice to users. A slightly richer communication not only better conveys the recommender's information to users but also enhances the recommender's learning by making users' choices more informative. Our results suggest that, in a wide range of environments, an optimal recommender system employs communication that is slightly more complex than a straightforward recommendation.\footnote{We do not model the communication cost explicitly because its shape, structure, and magnitude critically depend on the applications. Instead, we characterize the tradeoff by evaluating the benefits of enriching communication.}

An illustrative example is app-based car navigation systems (e.g., Google Maps or, Waze). In recent years, such navigation apps have become extremely popular.\footnote{According to \citet{Android2019}, Google Maps became the second app (after YouTube) to reach five billion downloads.} Navigation apps have an immense information advantage over individual drivers because they use aggregated information to dynamically detect traffic jams and then recommend less-congested routes. Accordingly, such apps are useful even for expert drivers who can figure out the shortest route without the recommender's help.

When a navigation app is launched, the app does not have complete information about road characteristics---local drivers have more comprehensive knowledge about their neighborhoods. For example, the app may miss information about hazard conditions associated with specific roads (e.g., high-crime-rate areas, rock-fall hazard zones, and accident blind spots). Such hazardous roads are often vacant because local drivers avoid them, leading a na{\"\i}ve recommender to consider such a route desirable and recommend it. Drivers unfamiliar with this hazard information might then follow the recommendation, exposing them to danger. To avoid this tragedy, the app must learn road characteristics to understand \emph{why} the road is vacant.

The classical rating-based approach is unsuitable for detecting hazards in the car navigation problem because (i) detailed ratings and reviews are often unavailable, and (ii) the app should not wait until it observes low payoffs because that would mean incidents or accidents indeed occur, causing problems for some users. Moreover, this problem cannot be solved completely by inputting hazard information manually because it is difficult to list all relevant hazard conditions in advance.\footnote{Nevertheless, navigation apps attempt to avoid this problem by manually inputting hazard information in practice. For example, in Israel and Brazil, Waze provides the option of alerting about high-risk routes: \url{https://support.google.com/waze/answer/7077122?hl=en} (seen on July 22, 2021).}

Our deviation-based learning approach solves this dilemma by extracting local drivers' (i.e., experts') knowledge. For example, when a hazardous route is recommended, a local driver ignores the recommendation and chooses a different route. Given that the app has an information advantage (i.e., insight into road congestion), such a decision would not be made unless the app has misunderstood something about the static map (with which the local driver is very familiar). Thus, upon observing a deviation, the app can update its knowledge about the static map. Conversely, if the app recommends a route that involves a potentially hazardous road but observes that the local driver followed the suggested route, then the app can conclude that the road is not so dangerous. In this manner, the app can better understand the static map and improve its recommendations. Furthermore, the deviation-based learning approach can detect hazardous roads \emph{before} additional incidents occur because the recommender can observe that local drivers avoid hazardous roads from the outset.

We analyze how the recommender can efficiently perform deviation-based learning. Formally, we analyze a stylized model in which each user has two arms (actions), as in seminal papers on information design theory (e.g., \citealt{KMP2014_wisdom} and \citealt{Che2017}). A user's payoff from an arm is normalized to zero, and his payoff from another arm is given by $x\theta + z$. The \emph{context} $x$ specifies the user's problem (in the navigation problem, a context includes elements such as the origin, destination, and means of transportation). We assume each user is an expert who knows the parameter $\theta$ and can correctly interpret his context $x$ to predict the first term of his payoff, $x \theta$ (i.e., he knows the static map and can find the shortest safe route). The recommender has additional information about the value of $z$ (e.g., congestion), which is not observed by the user. We assume that local drivers are more knowledgeable than the recommender about the static map; the recommender does not at first know the parameter $\theta$ and must learn it over time. For each user, the recommender sends a recommendation (message) based on a precommitted information structure. Upon observing the recommendation, the user forms a belief about the unobservable payoff component $z$ and selects either one of the two actions.

We demonstrate that the size of the message space is crucial for efficiency, showing that by making the message space \emph{slightly} larger than the action space, we obtain a \emph{very large} welfare gain.
A large message space enables the recommender to send a signal that indicates the recommender is ``on the fence'' which means that the payoffs associated with the two distinct actions are likely similar.
The availability of such messaging reveals users' information more efficiently and improves the learning rate exponentially without sacrificing the utilization of current knowledge.

First, we consider a binary message space, which is the same size as the action space. We first analyze the \emph{straightforward policy}, which simply informs the user which arm is estimated to be better. Our first main theorem shows that learning is very slow under the straightforward policy, and therefore, users suffer from substantial welfare loss. Here, recall that the recommended arm is chosen based on the recommender's current knowledge. Given the recommender has an information advantage, provided the recommender knows the state \emph{moderately} well, users are prone to following the recommendation blindly despite its flaws. Because the recommender knows that no deviation will occur, she learns nothing from users' subsequent behaviors. Formally, we prove that the expected number of users required to improve the recommendations increases exponentially as the quality of the recommender's knowledge improves. This effect slows learning severely, which has a large welfare cost: While the per-round welfare loss in this situation is moderately small (because most users want to follow the recommendation blindly), the loss accumulates to a large amount in the long run.

We demonstrate that an ideal solution to the problem above is to use a ternary message space.
We focus on the \emph{ternary policy}, a simple policy that recommends a particular arm only if the recommender is confident in her prediction. Otherwise, the recommender explains that she is ``on the fence,'' which means that, based on the recommender's current information, the two actions are predicted to produce similar payoffs. When the recommender is confident about her prediction (which is almost always the case after the quality of her knowledge has become high), the user also confidently follows the recommendation, which maximizes the true payoff with high probability. Furthermore, when the recommender admits that she is on the fence, the user's choice is very useful in updating the recommender's belief: The user's choice reveals whether the recommender overestimates or underestimates the state, and this information shrinks the recommender's confidence interval geometrically. With the ternary message space, the total welfare loss is bounded by a constant (independent of the number of users). We confirm this theoretical result by conducting numerical simulation and demonstrate that the ternary policy reduces the welfare loss by 99\% compared to the straightforward policy under a certain simulation setting. Note also that the performance difference becomes arbitrarily large when we consider a longer time horizon. Accordingly, the recommender can improve the learning rate and social welfare drastically by increasing the size of the message space just by one.

To confirm the superiority of the ternary policy, we also develop and analyze two further binary policies, the \emph{myopic policy} and the \emph{exploration-versus-exploitation (EvE) policy}. The myopic policy maximizes the current user's expected payoff with respect to the recommender's current knowledge. While the myopic policy sometimes achieves a strictly better payoff than the straightforward policy, it is asymptotically equivalent to the straightforward policy and the order of welfare loss is also the same. The EvE policy sacrifices early users' payoffs but rapidly learns the state at first and exploits the knowledge gained to achieve better welfare for late users. Among the three binary policies, the EvE policy performs the best. The myopic policy and the EvE policy feature several drawbacks and are difficult to implement. The ternary policy is easier to use, despite requiring one more message to be sent. Moreover, we demonstrate that the ternary policy substantially outperforms all three binary policies in terms of social welfare.

The rest of the paper is organized as follows. Section~\ref{sec: related literature} reviews the literature. Section~\ref{sec: model} describes the model. Section~\ref{sec_binary} studies the straightforward policy. Section~\ref{sec_ternary} studies the ternary policy. Section~\ref{sec: other binary policies} considers the myopic policy and the EvE policy. Section~\ref{sec_sim} presents the simulation results. Section~\ref{sec: conclusion} concludes the research.

\section{Related Literature}\label{sec: related literature}

\paragraph*{Information Design}

The literature on strategic experimentation \citep[e.g.,][]{bolton1999strategic,KMP2014_wisdom,Che2017} has considered an environment where a social planner can improve (utilitarian) social welfare by inducing early users' effort for exploration, while myopic users have no incentive to explore the state. The previous studies have demonstrated that effort for exploration can be induced by controlling users' information.
In our recommender's problem, the recommender also wants to explore information to improve the payoffs of late users. However, to achieve this, we sacrifice no user's payoff: By increasing the message space slightly, we can improve all users' payoffs substantially. Rather, this paper points out that there is a tradeoff between choice efficiency and communication complexity.

Furthermore, this paper elucidates how the recommender learns experts' knowledge via users' actions. This contrasts with previous studies on strategic experimentation and information design \citep[e.g.,][]{KG2011,bergemann2016bayes,bergemann2016information}, which have explored ways of incentivizing agents to obey recommendations. Indeed, when either (i) the recommender (sender) has complete information about the underlying parameter (as in information design models) or (ii) payoffs (or signals about them) are observable (as in strategic experimentation models), a version of the ``revelation principle'' (originally introduced by \citealp{myerson1982optimal}) holds. In these cases, without loss of generality, we can focus on incentive-compatible straightforward policies (which always recommend actions from which no user has an incentive to deviate). By contrast, we demonstrate that when the recommender learns about underlying parameters by observing how users act after receiving the recommendation, only recommending a choice is often inefficient.

\paragraph*{Recommender System}
Although the recommender systems have mostly focused on predicting ratings, the vulnerability of rating-based learning has been widely recognized. \citet{salganik2006experimental} and \citet{muchnik2013social} show that prior ratings bias the evaluations of subsequent reviewers.
\citet{marlin2009collaborative} show that ratings often involve nonrandom missing data because users choose which item to rate.
\citet{mayzlin2014promotional} and \citet{luca2016fake} report that firms attempt to manipulate their online reputations strategically. While the literature has proposed several approaches to addressing these issues (for example, \citet{sgr2016_nipsfeedbackloop} propose a way to correct bias by formulating recommendations as a control problem), the solutions proposed thus far remain somewhat heuristic. That is, their authors have not identified the fundamental source of the biases in rating systems using a model featuring rational agents.\footnote{See the survey of the biases in rating systems by \citet{CDWFWH_biasrecommend}.}
By contrast, our deviation-based approach is fundamentally free from these biases because our approach does not assume the availability of ratings.

\paragraph*{Learning from Observed Behaviors}

In the literature of economic theory, inferring a rational agent's preferences given their observed choices is rather a classic question \citep[\emph{revealed preference theory}, pioneered by][]{samuelson1938note}.\footnote{More recently, \cite{Cheung2021} has developed a revealed-preference framework under the presence of recommendations and proposed a method for identifying how recommendations influence decisions.}
Furthermore, recent studies on machine learning and operations research, including inverse reinforcement learning \citep{ng2000} and contextual inverse optimization \citep{ahuja2001,omar2021} have also proposed learning methods that recover a decision maker's objective function from his behavior.\footnote{Classical learning methods, such as reinforcement learning (see \citealt{Sutton2018}, a standard textbook on this subject) and algorithms that solve multi-armed bandit problems \citep{Thompson1933,lai1985}, assume that the learner can directly observe realized payoffs.} These methods can usefully extract experts' knowledge to make better predictions about users' payoffs.

Our contribution to this literature can be summarized as follows.
First, we elucidate the effect of the recommender's information advantage. In many real-world problems (e.g., navigation), the recommender is not informationally dominated by expert users; thus, decisions made by experts who are not informed of the recommender's information are typically suboptimal. This paper proposes a method to efficiently extract experts' knowledge and combine it with the recommender's own information.
Second, we articulate the role of users' beliefs about the accuracy of the recommender's predictions. When the recommendation is accurate, users tend to follow recommendations blindly, and therefore, learning stalls under a na{\"\i}ve policy.
Third, we demonstrate that the recommender can improve her learning rate significantly by intervening in the data generation process through information design. In our environment, learning under the ternary policy is exponentially faster than learning under the binary (straightforward) policy. The difference in social welfare achieved is also large.

The marketing science literature has proposed adaptive conjoint analysis as a method of posing questions to estimate users' preference parameters in an adaptive manner. Several studies, such as \cite{toubia2007} and \cite{saure2019ellipsoidal}, have considered adaptive \emph{choice-based} conjoint analysis, which regards choice sets as questions and actual choices as answers to those questions. This strand of the literature has also developed efficient methods for intervening in the data generation process to extract users' knowledge. However, in the recommender problem, the recommender is not allowed to select users' choice sets to elicit their preferences.

\section{Model}\label{sec: model}

\subsection{Environment}

We consider a sequential game that involves a long-lived \emph{recommender} and $T$ short-lived \emph{user}s. Initially, the state of the world $\theta \sim \Unif[-1, 1]$ is drawn. We assume that all users are experts and more knowledgeable than the recommender about the state $\theta$ initially.\footnote{As long as the recommender can identify the set of expert users, she can exclude nonexpert users from the model. In the navigation app example, it should not be difficult for the app to identify the set of local residents who drive cars frequently. Once the recommender trains the system using the data of the experts' decisions, then she can use it to make recommendations to nonexpert users.} Formally, we assume that while users know the realization of $\theta$, the recommender knows only the distribution of $\theta$. Accordingly, the recommender learns about $\theta$ via the data obtained. 

Users arrive sequentially. At the beginning of round $t \in [T] \coloneqq \{1,\ldots, T\}$, user
$t$ arrives with the shared \emph{context} $x_t \sim \Normal$, where $\Normal$ is the standard (i.e., with a zero mean and unit variance) normal distribution.\footnote{We assume that the state $\theta$ and context $x_t$ are one-dimensional because this assumption enables us to write the recommender's estimate as a tractable closed-form formula ($\Ex_t[\theta] = m_t \coloneqq (u_t + l_t)/2$), where $(u_t, l_t)$ is defined in page~\pageref{sentence: l_t u_t defined}). If $\theta$ and $x_t$ are multi-dimensional, then $\Ex_t[\theta]$ is a centroid of a convex polytope defined by $t-1$ faces, which does not have a tractable formula and is generally \#P-hard to compute \citep{Radamacher2007convexpolytope}, while a reasonable approximation is achieved by a random sampling method \citep{bertsimas2004}. We consider the high-level conclusion of this paper does not crucially depend on the dimensionality of the state and contexts.}
The context $x_t$ is public information and observed by both user $t$ and the recommender. The context specifies the user's decision problem. The recommender additionally observes her private information $z_t \sim \Normal$, the realization of which is not disclosed to user $t$. Each user has binary actions available: arm $-1$ and arm $1$.\footnote{Alternatively, we can assume that each user has many actions but all except two are obviously undesirable in each round.}
Without loss of generality, the user's payoff from choosing arm $-1$ is normalized to zero: $r_t(-1) = 0$.\footnote{\label{footnote: normalization} We are not assuming that arm $-1$ is a safe arm, but normalizing the payoff of one of the two arms. To illustrate this, let us start from the following formulation: $r_t(1) = x_t\theta^{(1)} + z_t^{(1)}$ and $r_t(-1) = x_t \theta^{(-1)} + z_t^{(-1)}$. The user chooses arm $1$ if and only if $r_t(1) > r_t(-1)$, i.e., $x_t(\theta^{(1)} - \theta^{(-1)}) + (z_t^{(1)} - z_t^{(-1)})$. By redefining $r_t(-1) \equiv 0$, $\theta = \theta^{(1)} - \theta^{(-1)}$ and $z_t = z_t^{(1)} - z_t^{(-1)}$, the model is reduced to a normalized one, without changing the users' decision problem.} The payoff from choosing arm $1$ is given by
\begin{equation} 
r_{t}(1) = x_t \theta + z_t.
\end{equation}
We refer to $x_t \theta$ as the \emph{static payoff} and $z_t$ as the \emph{dynamic payoff}. These names come from the navigation problem presented as an illustrative example, in which users are assumed to be familiar with the static road map but do not observe dynamic congestion information before they select the route. All the variables, $\theta$, $(x_t)_{t \in [T]}$, $(z_t)_{t \in [T]}$ are drawn independently of each other.

In round $t$, the recommender first selects a \emph{recommendation} $a_t \in A$, where $A$ is the \emph{message space}. For example, if the recommender simply reports the estimated-to-be-better arm, then the message space is equal to the action space: $A = \{-1, 1\}$. Observing the recommendation $a_t$, user $t$ forms a posterior belief about the realization of $z_t$ and chooses an action $b_t \in B = \{-1, 1\}$. User $t$ receives a payoff of $r_t(b_t)$ and leaves the market. The recommender cannot observe users' payoffs.

Technically, by sending a signal, the recommender informs the realization of the dynamic payoff $z_t$, which users cannot observe directly. In round $t$, the recommender initially commits to an \emph{signal function} $\mu_t:\Real \to A$ that maps a dynamic payoff $z_t$ to a message $a_t$. Subsequently, the recommender observes the realization of $z_t$ and mechanically submits a message (recommendation) $a_t = \mu_t(z_t)$. When the recommender chooses the round-$t$ information structure, she can observe the sequences of all contexts $(x_s)_{s=1}^t$, all past dynamic payoffs $(z_s)_{s=1}^{t-1}$, all past messages, $(a_s)_{s=1}^{t-1}$, and all past actions that users took, $(b_s)_{s=1}^{t-1}$. A \emph{policy} is a rule to map the information that the recommender observes ($(z_s,a_s, b_s)_{s=1}^{t-1}$ and $(x_s)_{s=1}^{t}$) to a signal function. 

Receiving a message $a_t = \mu_t(z_t)$, user $t$ forms a posterior belief about $z_t$, and choose an arm that has a better conditional expected payoff. As in the information design literature \citep[e.g.,][]{KG2011}, we assume that the user knows the information structure: User $t$ observes the signal function $\mu_t$.\footnote{The other information is redundant for the user's decision problem, given that user $t$ observes the signal function $\mu_t$.}
Upon observing $a_t$, user $t$ forms his posterior belief about the dynamic payoff $z_t$. User $t$ computes the conditional expected payoff of arm $1$, $\Ex_{z_t}[x_t \theta + z_t|\mu_t, a_t]$, based on his posterior belief. Then, user $t$ selects an arm $b_t \in B \coloneqq \{-1, 1\}$, which is expected to provide a larger payoff: $b_t = 1$ if $\Ex_{z_t}[x_t \theta + z_t|\mu_t, a_t] > 0$ and $b_t = -1$ otherwise.

\subsection{Regret}
Utilitarian \emph{social welfare} is defined as the sum of all users' payoffs: $\sum_{t=1}^T r_t(b_t)$.
However, its absolute value is meaningless because we normalize $r_t(-1) \equiv 0$.\footnote{See also footnote~\ref{footnote: normalization}.}
Instead, we quantify welfare loss in comparison to the first-best scenario, which is invariant to the normalization. We define per-round regret, $\regret$, and (cumulative) regret, $\Regret$, as follows:
\begin{align}
\regret(t) &\coloneqq r_{t}(\ist) - r_{t}(b_t); \\
\Regret(T) &\coloneqq \sum_{t=1}^T \regret(t),
\end{align}
where $\ist \coloneqq \argmax_{b \in \{-1,1\}} r_{t}(b)$ is the superior arm with respect to true payoffs. 
Per-round regret, $\regret(t)$, represents the loss of the current (round-$t$) user due to a suboptimal choice. While user $t$ could enjoy $r_t(\ist)$ if he were to observe $z_t$, his actual payoff is $r_t(b_t)$. Therefore, his loss compared with the (unattainable) first-best case is given by $\regret(t)$. Since $\Regret(T)$ is a summation of $\regret(t)$, $\Regret(T)$ is the difference between the total payoffs from best arms $\sum_{t=1}^T r_t(\ist)$ (which is unattainable) and the actual total payoffs $\sum_{t=1}^T r_t(b_t)$.
The first-best benchmark, $\sum_{t=1}^T r_t(\ist)$, is independent of the policy and users' choices. Thus, the maximization of the total payoffs is equivalent to the minimization of the regret.

If the recommender already knows (or has accurately learned) the state $\theta$, then the recommender would always inform user $t$ of the superior arm, and the user would always obey the recommendation. Therefore, $b_t = \ist$ and $\regret(t) = 0$ would be achieved.
Conversely, if the recommender's belief about the state $\theta$ is inaccurate, then users cannot always select the superior arm.
Therefore, regret also measures the progress of the recommender's learning of $\theta$.

In this paper, we characterize the relationship between the size of the message space $|A|$ and the order of regret $\Regret(T)$. When the message space is a singleton (i.e., $|A| = 1$), the recommender can deliver no information about the dynamic payoff component $z_t$.
Consequently, users suffer from constant welfare loss in each round; therefore, the regret grows linearly in $T$, that is, $\Regret(T) = \Theta(T)$. By contrast, if the message space is a continuum (i.e., $A = \Real$), the recommender can inform each user $t$ of the ``raw data'' about the dynamic payoff $z_t$, i.e., she can send $a_t = z_t$ as a message.
In this case, users can recover true payoffs $r_t(1)$ and select the superior arms for every round. There is no need for the recommender to learn, and the regret of exactly zero is achieved, that is, $\Regret(T) = 0$ for all $T$. Nevertheless, an infinite message space incurs a large communication cost, making it considerably inconvenient. Practically, it is infeasible for real-world recommender systems to disclose all current congestion information.
The above argument indicates that there is a tradeoff between regret and communication complexity (i.e., the size of the message space $|A|$), which remains to be evaluated.
The following sections characterize the regret incurred by small finite message spaces, namely, the cases of binary and ternary message spaces ($|A| = 2, 3$).

\section{Binary Straightforward Policy}\label{sec_binary}

\subsection{Policy}

First, we consider the case of the binary message space, i.e., $A = \{-1, 1\} = B$. We say that a policy is \emph{binary} if it employs a binary message space. We begin with a \emph{straightforward policy} that simply discloses an arm that the recommender estimates to be superior.

The recommender's estimation proceeds as follows.
After observing user $t$'s choice $b_t$, the recommender updates the posterior distribution about $\theta$, characterized by $(l_t, u_t)$, according to Bayes' rule.\label{sentence: l_t u_t defined} 
The recommender's belief at the beginning of round $1$ is the same as the prior belief: $\Unif[-1, 1]$. Due to the property of uniform distributions, the posterior distribution of $\theta$ always belongs to the class of uniform distributions. 
The posterior distribution at the beginning of round $t$ is specified by $\Unif[l_t, u_t]$, where $l_t$ and $u_t$ are the lower and upper bounds, respectively, of the confidence interval at the beginning of round $t$.
Note that the \emph{confidence interval} $[l_t, u_t]$ shrinks over time:
\begin{equation}
    -1 \eqqcolon l_1 \le l_2 \le \cdots \le l_{T-1} \le l_{T} \le \theta \le u_T \le u_{T-1} \le \cdots \le u_2 \le u_1 \coloneqq 1,
\end{equation}
and thus the \emph{width of the confidence interval} $w_t \coloneqq u_t - l_t$ is monotonically decreasing.
In round $t$, the recommender believes that $\theta$ is drawn from the posterior distribution, $\Unif[l_t, u_t]$, and therefore, the estimated payoff from arm $1$ is
\begin{equation} 
 \hatr_{t}(1) \coloneqq \Ex_{\tiltheta_t \sim \Unif[l_t, u_t]}\left[x_t \tiltheta_t\right] + z_t 
 = x_t m_t + z_t,
\end{equation}
where $m_t \coloneqq (l_t + u_t)/2 = \Ex_{\tiltheta_t \sim \Unif[l_t, u_t]}[\tiltheta_t]$.

The straightforward policy recommends arm $1$ if and only if the recommender believes that the expected payoff from arm $1$ is larger than that from arm $-1$, i.e., $\hatr_{t}(1) = x_t m_t + z_t > 0 = \hatr_t(-1)$.\footnote{We ignore equalities of continuous variables that are of measure zero, such as $\hatr_{t}(1) = 0$.} That is, the signal function $\mu_t$ is given by
\begin{equation}
\mu_t(z_t; m_t, x_t) = 
\left\{
\begin{array}{ll} 
1 & \text{ if }x_t m_t + z_t > 0; \\
-1 & \text{ otherwise}.
\end{array}
\right.
\end{equation}

Although the straightforward policy is simple, natural, intuitive, and easy to understand, it is not an optimal binary policy. Indeed, Section~\ref{sec: other binary policies} introduces two binary policies that are substantially more complex and perform better than the straightforward policy. Nevertheless, we discuss the straightforward policy as the main benchmark because its simple and natural structure is desirable in terms of communication complexity. For further discussion, see Section~\ref{sec: other binary policies}.

\subsection{Learning}

From now, we consider users' action choices under the straightforward policy. User $t$'s conditional expected payoff from choosing arm $1$ is given by
\begin{equation}
    \Ex[r_t(1)|\mu_t, a_t] = x_t \theta + Z_t,
\end{equation}
where $Z_t \coloneqq \Ex[z_t|\mu_t, a_t]$.

\begin{figure}
    \centering
    \includegraphics[width = 0.7 \textwidth]{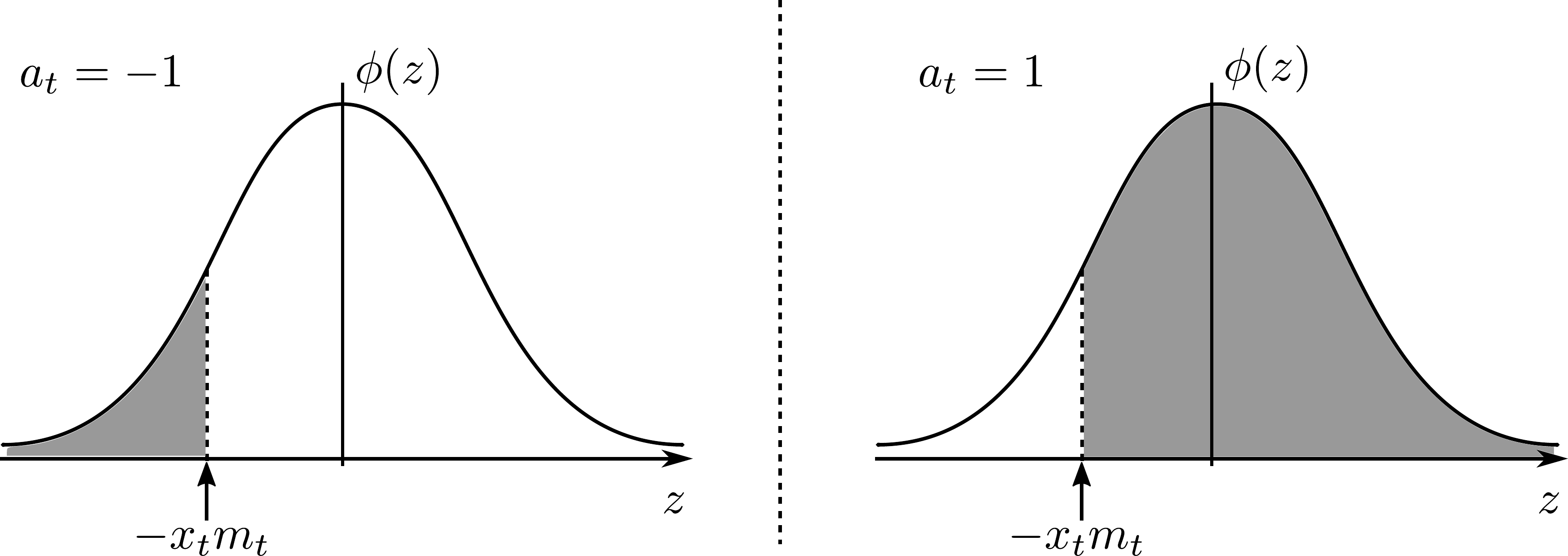}
    \caption{The shape of the posterior distribution of $z_t$ conditional on $a_t=-1$ (left) and $a_t = 1$ (right) being sent under the straightforward policy, when $x_t m_t > 0$.}
    \label{fig:ez_recom}
\end{figure}

The prior distribution of $z_t$ is the standard normal distribution, $\Normal$. In addition, $a_t = 1$ implies $z_t > -x_t m_t$, whereas $a_t = -1$ implies $z_t < -x_t m_t$. Accordingly, the posterior distribution of $z_t$ is always a truncated standard normal distribution. 
Let $\trunkN(\alpha, \beta)$ be the truncated standard normal distribution with support $(\alpha, \beta)$. Then, the posterior distribution of $z_t$ after $a_t = 1$ and $a_t = -1$ are $\trunkN(-x_t m_t, \infty)$ and $\trunkN(-\infty, -x_t m_t)$, respectively. These distributions are illustrated as Figure~\ref{fig:ez_recom}. To summarize, we have
\begin{equation}\label{ineq_ehatzt}
    \Ehatzt \coloneqq \Ex[z_t|\mu_t, a_t] = 
\left\{
\begin{array}{ll} 
\Ex_{z \sim \trunkN(-\infty, -x_t m_t)} [z] & \text{ if }a_t=-1; \\
\Ex_{z \sim \trunkN(-x_t m_t, \infty)} [z] & \text{ if }a_t=1.
\end{array}
\right.
\end{equation}
The arm that user $t$ will choose is as follows:
\begin{equation} \label{eq_user's decision}
b_t = 
\left\{
\begin{array}{ll} 
1 & \text{ if }x_t \theta + Z_t > 0; \\
-1 & \text{ otherwise}.
\end{array}
\right.
\end{equation}

Upon observing the user's decision $b_t$, the recommender updates her confidence interval, $[l_t, u_t]$. When the user chooses $b_t = 1$, the recommender can recognize that $x_t \theta + Z_t > 0$. If $x_t > 0$, then this is equivalent to $\theta > -Z_t/x_t$; and if $x_t < 0$, then this is equivalent to $\theta < -Z_t/x_t$. Using this information, the recommender may be able to shrink the support of the posterior distribution of $\theta$. We can analyze the case of $b_t = -1$ in a similar manner. The belief update rule is as follows:
\begin{equation}\label{ineq_update_left}
    l_{t+1} = 
    \left\{
    \begin{array}{ll} 
    l_t & \text{ if } b_t \cdot \sgn(x_t) < 0; \\
    \max\{l_t, - Z_t/x_t\} & \text{ if } b_t \cdot \sgn(x_t) > 0,
\end{array}
\right.
\end{equation}
\begin{equation}\label{ineq_update_right}
    u_{t+1} = 
    \left\{
    \begin{array}{ll} 
    \min\{u_t, -Z_t/x_t\} & \text{ if } b_t \cdot \sgn(x_t) < 0; \\
    u_t & \text{ if } b_t \cdot \sgn(x_t) > 0,
\end{array}
\right.
\end{equation}
where $\sgn$ is the following signum function:\footnote{Because $x_t = 0$ occurs with probability zero, we ignore such a realization.}
\begin{equation}
    \sgn(x) \coloneqq
    \left\{
    \begin{array}{ll}
    1 & \text{ if } x > 0; \\
    -1 & \text{ if } x < 0.
    \end{array}
    \right.
\end{equation}

\subsection{Failure}

We present our first main theorem, which evaluates the order of total regret under the straightforward policy.

\begin{thm}[Regret Bound of Straightforward Policy]\label{thm_binary}
For the straightforward policy, there exists a $\tilTheta(1)$ (polylogarithmic) function\footnote{$\tilO, \tilOmega$, and $\tilTheta$ are Landau notations that ignore polylogarithmic factors (e.g., $\tilTheta(\sqrt{T}) = (\log T)^c \Theta(\sqrt{T})$ for some $c\in\Real$). We often treat these factors as if they were constant because polylogarithmic factors grow very slowly ($o(N^\epsilon)$ for any exponent $\epsilon > 0$).} $f:\mathbb{Z} \to \Real$ such that 
\begin{equation}
\Ex[\Regret(T)] \ge T/f(T).
\end{equation}
\end{thm}
Theorem~\ref{thm_binary} shows that the total regret is $\tilOmega(T)$, which implies that users suffer from a large per-round regret even in the long run.

All the formal proofs are presented in Appendix~\ref{sec: proof}. The intuition of Theorem~\ref{thm_binary} is as follows. While each user precisely knows his static payoff $x_t \theta$, he has access to the dynamic payoff $z_t$ only via recommendation. To help the user make the best decision, the recommender must identify which arm is better as a whole. The recommender must therefore learn about the state $\theta$ in order to figure out the value of $r_t(1) = x_t \theta + z_t$ via the users' feedback $b_t$. As the recommender becomes more knowledgeable about $\theta$, users' feedback becomes less informative: Rational users rarely deviate from (moderately) accurate recommendations because the recommender's information advantage (in terms of information about the dynamic payoff term) tends to dominate the estimation error. Consequently, when recommendations are accurate, deviations are rarely observed, and the recommender has few opportunities to improve her estimations.

In the following, we provide two lemmas that characterize the problem and then discuss how we derive Theorem~\ref{thm_binary} from these lemmas.

\begin{lem}[Lower Bound on Regret per Round]\label{lem_reglower_round}
Under the straightforward policy, there exists a universal constant $\Creg > 0$ such that the following inequality holds:\footnote{A universal constant is a value that does not depend on any model parameters.}
\begin{equation}\label{ineq_reglower_round}
\Ex[\regret(t)]
\ge \Creg |\theta - m_t|^2.
\end{equation}
\end{lem}

Since the recommender does not know $\theta$, she substitutes $m_t$ for $\theta$ to determine her recommendation.
The probability that the recommender fails to recommend the superior arm is proportional to $|\theta - m_t|$, and the welfare cost from such an event is also proportional to $|\theta - m_t|$. Accordingly, the per-round expected regret is at the rate of $\Omega(|\theta - m_t|^2)$. Note that, from the perspective of the recommender, the posterior distribution of $\theta$ is $\Unif[l_t, u_t]$, and therefore, the conditional expectation of $|\theta - m_t|^2$ is $\Theta(w_t^2)$.

\begin{lem}[Upper Bound on Probability of Update]\label{lem_update} 
Under the straightforward policy, there exists a universal constant $\Cupdate > 0$ such that,
for all $w_t \le \Cupdate$, 
\begin{equation}
\Prob[(l_{t+1}, u_{t+1}) \ne (l_t, u_t)] \le
\exp\left(-\frac{\Cupdate}{w_t}\right).
\end{equation}
\end{lem}

User $t$ compares two factors when making his decision: (i) the recommender's estimation error of the static payoff term $|x_t (\theta - m_t)|$ and (ii) the recommender's information advantage about the dynamic payoff term $z_t$.
When the former term is small, the user blindly obeys the recommendation, and the user's decision does not provide additional information. Because $w_t > |\theta - m_t|$, the former factor is bounded by $|x_t w_t|$. For a user's decision to be informative, $|x_t|$ must be $\Omega(1/w_t)$ (in which case $|x_t (\theta - m_t)|$ exceeds a threshold value). Because $x_t$ follows a normal distribution, the probability of such a context decreases exponentially in $1/w_t$.\footnote{A similar result holds whenever $x_t$ follows a sub-Gaussian distribution, where the probability of observing $x_t$ decays at an exponential rate with respect to $|x_t|$.
Conversely, when the distribution of $x_t$ is heavy-tailed, the conclusion of Lemma~\ref{lem_update} may not hold.}

Lemma~\ref{lem_update} states that the recommender's learning stalls when $w_t$ is moderately small. In particular, if $w_t = 2\Cupdate/(\log T) = \Theta(1/(\log T))$, then the probability of her belief update is $1/T^2$. This implies that no update occurs in the next $T$ rounds with a probability of at least $1 - 1/T$.

We use these lemmas to obtain the total regret bound presented in Theorem~\ref{thm_binary}. First, Lemma~\ref{lem_update} implies that the update of $\theta$ is likely to stall when it reaches $w_t = |\theta-m_t| = \Theta(1/(\log T))$.
Given $|\theta - m_t| = \Theta(1/(\log T))$, Lemma~\ref{lem_reglower_round} implies that the per-round (expected) regret is $\Theta(1/(\log T)^2)$. 
Consequently, the order of total regret is $\Omega(T/(\log T)^2) = \tilde{\Omega}(T)$, implying that users suffer from large per-round regrets even in the long run. Thus, we obtain the regret bound presented as Theorem~\ref{thm_binary}.

\section{Ternary Policy}\label{sec_ternary}

\subsection{Policy}

In Section~\ref{sec_binary}, we demonstrate that the straightforward policy suffers from an approximately linear regret. This section shows that the regret rate improves substantially if we introduce a ternary message space, $A = \{-1, 0, 1\} = B \cup \{0\}$, and incorporate the additional message, ``0'', into the straightforward policy in a simple manner.
The ternary message space allows the recommender to inform users that she is ``on the fence.'' When the recommender is confident in her recommendation, she sends either $a_t = -1$ or $a_t = 1$. If the recommender predicts that the user should be approximately indifferent to the choice between the two arms, then she sends $a_t = 0$ instead.

Specifically, we introduce a sequence of parameters $(\eps_t)_{t=1}^T$, where $\eps_t > 0$ for all $t \in [T]$. This specifies whether the recommender is confident in her prediction. If $\hatr_t(1) > \eps_t$, then the recommender is confident about the superiority of arm $1$, and therefore recommends arm $1$: $a_t = 1$. Conversely, if $\hatr_t(1) < - \eps_t$, then the recommender is confident about the superiority of arm $-1$, and therefore recommends arm $-1$: $a_t = -1$. In the third case, i.e., $-\eps_t < \hatr_t(1) < \eps_t$, the recommender states honestly that she is on the fence; she sends the message $a_t = 0$, implying that she predicts similar payoffs for arms $1$ and $-1$. This policy is summarized as follows:
\begin{equation}
\mu_t(z_t; m_t, x_t, \epsilon_t) = 
\left\{
\begin{array}{ll} 
1 & \text{ if }x_t m_t + z_t > \eps_t; \\
0 & \text{ if } \eps_t > x_t m_t + z_t > - \eps_t; \\
-1 & \text{ if }x_t m_t + z_t < -\eps_t.
\end{array}
\right.
\end{equation}
Because this paper does not discuss any other policy that employs the ternary message space, we refer to this as the \emph{ternary policy}.

User $t$'s posterior belief about $z_t$ is given by (i) $z_t \sim \trunkN(-\infty, -x_t m_t-\eps_t)$ given $a_t = -1$; (ii) $z_t \sim \trunkN(-x_t m_t-\eps_t, -x_t m_t+\eps_t)$ given $a_t = 0$; and (iii) $z_t \sim \trunkN(-x_t m_t+\eps_t, \infty)$ given $a_t = 1$. Accordingly, the conditional expectation of $z_t$ with respect to the posterior distribution is formulated as follows (and illustrated in Figure~\ref{fig:ez_recom_three}).
\begin{equation}\label{ineq_ehatzt_ternary}
    \Ehatzt \coloneqq \Ex[z_t|\mu_t, a_t] = 
\left\{
\begin{array}{ll} 
\Ex_{z \sim \trunkN(-\infty, -x_t m_t-\eps_t)} [z] & \text{ if }a_t=-1; \\
\Ex_{z \sim \trunkN(-x_t m_t-\eps_t, -x_t m_t+\eps_t)} [z] & \text{ if }a_t=0; \\
\Ex_{z \sim \trunkN(-x_t m_t+\eps_t, \infty)} [z] & \text{ if }a_t=1.
\end{array}
\right.
\end{equation}
Given the new specifications of $a_t$ and $Z_t$, the user's decision rule for choosing $b_t$ (given in Eq.~\eqref{eq_user's decision}) and the belief update rule for deciding $(l_{t+1}, u_{t+1})$ (given in Eq.~\eqref{ineq_update_left} and \eqref{ineq_update_right}) do not change.

\begin{figure}
	\includegraphics[width = \textwidth]{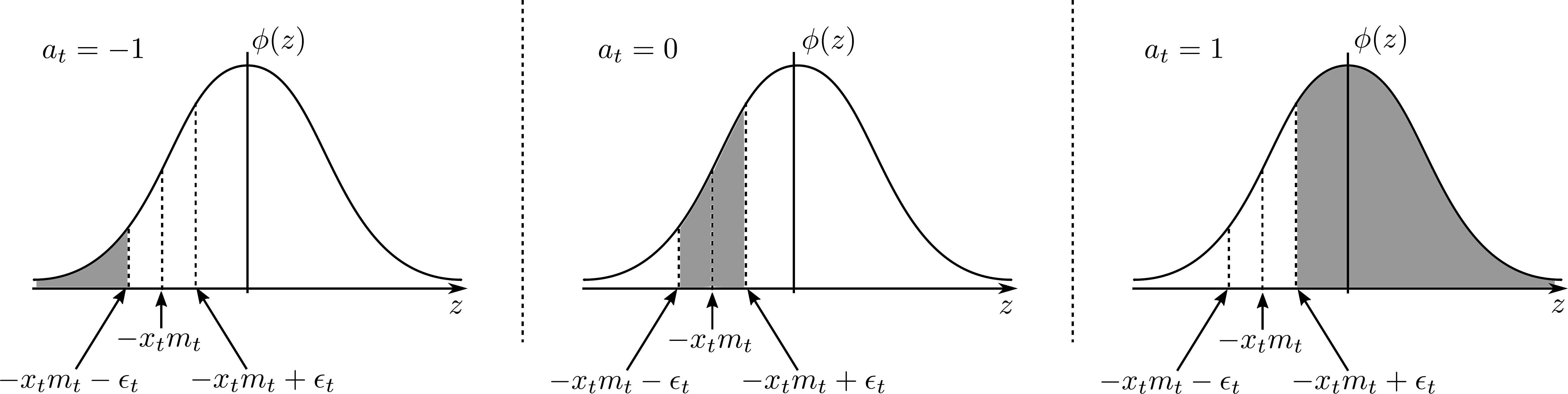}
	\caption{The shape of the posterior distribution of $z_t$ conditional on $a_t=-1$ (left), $a_t = 0$, and $a_t = 1$ (right) being sent under the ternary policy, when $x_t m_t > 0$.}
	\label{fig:ez_recom_three}
\end{figure}

\subsection{Success of Ternary Policy}

The following theorem characterizes the total regret achieved by the ternary policy.
\begin{thm}[Regret Bound of Ternary Policy]\label{thm_three}
Let $\eps_t = \Ceps w_t$ where $\Ceps > 0$ is an arbitrary constant. 
Then, under the ternary policy, there exists a constant $\Cternary>0$ that depends only on $\Ceps$ and with which the regret is bounded as:
\begin{equation}
\Ex[\Regret(T)] \le \Cternary.
\end{equation}
\end{thm}
Under the ternary policy, the expected regret is $O(1)$, which is the best possible regret order and is often celebrated in theoretical computer science. This result contrasts with Theorem~\ref{thm_binary}, which indicates that the straightforward policy suffers from $\tilOmega(T)$ regret. Theorem~\ref{thm_three} implies that a slight expansion of the message space drastically improves the regret rate.

The ternary message space benefits users in two ways. First, it increases the current user's payoff by enabling more informative signaling about $z_t$.
Because a more informed conditional expectation of $z_t$ closely approximates the realization, a larger message space enables users to make better decision.\footnote{Since the ternary policy does not dominate the straightforward policy in terms of the Blackwell informativeness criterion \citep{blackwell1953equivalent}, the ternary policy may provide a smaller expected payoff for current users. This happens if $\eps_t$ is too large. Meanwhile, Section~\ref{sec_sim} demonstrates that for $\Ceps = 1/4$, all users receive larger expected payoffs under the ternary policy. Alternatively, we can consider a slightly different ternary policy that dominates the straightforward policy. For example, by defining $\mu_t(z_t) = 1$ if $x_t m_t + z_t > 0$, $\mu_t(z_t) = 0$ if $x_t m_t + z_t \in (- \epsilon_t, 0)$, and $\mu_t(z_t) = -1$ if $x_t m_t + z_t < - \epsilon_t$, the alternative ternary policy informationally dominates the straightforward policy and, therefore, provides the current user with a larger expected payoff than that offered by the straightforward policy, regardless of the choice of $\eps_t$.}
However, the first effect alone does not improve the order of regret.

The second effect is critical, positioning it as the basis of the proof of Theorem~\ref{thm_three}. Under the ternary policy, the learning rate improves drastically. This is because users' actions after they receive $a_t = 0$ are very informative for the recommender's belief update. This is because the ``on-the-fence'' message $a_t = 0$ is sent when the recommender is ``unconfident'' about the better arm, and a user's action tells the answer to the recommender. This effect allows us to prove that the recommender's confidence interval shrinks geometrically, leading the per-round regret to diminish exponentially. The following key lemma illustrates this fact.

\begin{lem}[Geometric Update]\label{lem_unknown_shrink}
Let $\eps_t = \Ceps w_t$ for $\Ceps > 0$. Then, under the ternary policy, there exists a constant $\Cgeomup>0$ that depends only on $\Ceps$ and with which the following inequality holds:
\begin{equation}\label{ineq_unknownrate}
\Prob\left[\left.w_{t+1} \le \frac{5}{6} w_t\right|a_t=0\right] \ge \Cgeomup.
\end{equation}
\end{lem}

The intuition for the geometric gain after $a_t = 0$ is as follows. When $\eps_t \approx 0$, upon observing $a_t = 0$, user $t$ can accurately figure out the realization of the dynamic payoff term: $z_t \approx - x_t m_t$. Given this, the user chooses $b_t = 1$ if $x_t \theta + z_t \approx x_t (\theta - m_t) > 0$ and $b_t = -1$ otherwise; in other words, by observing $b_t$, the recommender can identify whether or not $\theta > m_t$. Since $m_t$ is the median of the confidence interval $[l_t, u_t]$, this observation (approximately) halves the width of the confidence interval.

While a smaller $\eps_t$ results in a larger update after $a_t = 0$ is sent, we cannot set $\eps_t = 0$ because in that case the probability of sending $a_t = 0$ becomes zero. The policy parameter $\eps_t$ must be chosen to balance this trade-off. Lemma~\ref{lem_unknown_shrink} shows that $\eps_t = \Ceps w_t$ (for any $\Ceps > 0$) is an appropriate choice in the sense that it achieves a constant per-round probability of geometric updates. Accordingly, the ternary policy shrinks the confidence interval exponentially in the total number of rounds in which $a_t = 0$ is sent.

We introduce two more lemmas to illustrate the proof sketch of Theorem~\ref{thm_three}. The second lemma, Lemma~\ref{lem_probunknown}, computes the probability that $a_t = 0$ is sent.

\begin{lem}[Probability of $a_t=0$]\label{lem_probunknown}
Under the ternary policy, there exist universal constants $\CfenceL, \CfenceU > 0$ such that the following equality holds:\footnote{``OtF'' stands for ``on the fence.''}
\begin{equation}
\CfenceL \eps_t \le \Prob[a_t = 0] \le \CfenceU \eps_t,
\end{equation}
\end{lem}
Lemma~\ref{lem_probunknown} states that the probability that $a_t = 0$ is recommended is linear in $\eps_t$. This result immediately follows from the fact that (i) $z_t$ follows a standard normal distribution, and (ii) $a_t = 0$ is sent when $z_t \in (- x_t m_t - \eps_t, - x_t m_t + \eps_t)$.

The third lemma, Lemma~\ref{lem_regeps}, bounds the per-round regret, $\regret(t)$, using a quadratic function of the policy parameter, $\eps_t$, and the width of the confidence interval, $w_t$.

\begin{lem}[Upper Bound on Regret per Round]\label{lem_regeps}
Under the ternary policy with $\eps_t = \Ceps w_t$, there exists a constant $\Cregt > 0$ that only depends on $\Ceps$ such that the following inequality holds:
\begin{equation}
\Ex[\regret(t)]
\le \Cregt w_t^2.
\end{equation}
\end{lem}

When an arm is recommended (i.e., when $a_t \neq 0$), then we can apply essentially the same analysis as the straightforward policy to derive the per-round expected regret of $O(w_t^2)$. The message $a_t = 0$ is sent with probability $\Theta(\eps_t)$ (by Lemma~\ref{lem_probunknown}). Since $a_t = 0$ is sent only when the utility is (approximately) indifferent between two arms, the per-round regret is bounded by $\eps_t + w_t$ in this case. Hence, the per-round expected regret for this case is $O(\max\{\eps_t,w_t\}^2)$. When we choose $\eps_t = \Ceps w_t$, then the per-round regret becomes quadratic in $w_t$.

The proof outline of Theorem~\ref{thm_three} is as follows. By Lemma~\ref{lem_probunknown}, the probability of $a_t = 0$ is $\Theta(\eps_t) = \Theta(w_t)$. Together with Lemma~\ref{lem_unknown_shrink}, it follows that, in round $t$, with probability $\Theta(w_t)$, the width of confidence interval $w_t$ shrinks geometrically to $w_{t+1} = (5/6) w_t$ or smaller. This leads an exponential reduction of the confidence interval to the total number of users to which the recommender has sent $a_t = 0$. Finally, when $\eps_t = \Theta(w_t)$, Lemma~\ref{lem_regeps} ensures that the per-round regret is $O(w_t^2)$. Let us refer to an interval between two geometric intervals as an \emph{epoch}. Since a geometric update occurs with probability $\Theta(w_t)$, the expected number of rounds contained in one epoch is $\Theta(1/w_t)$. The regret incurred per round is $\Theta(w_t^2)$; thus, the total regret incurred in one epoch is $\Theta(w_t^2 \times 1/w_t) = \Theta(w_t)$. Accordingly, the expected regret associated with each epoch is bounded by a geometric sequence whose common ratio is $5/6 < 1$. The total regret is the sum of the regret from all the epochs. Accordingly, the total regret is bounded by the sum of a geometric series, which converges to a constant.

\section{Other Binary Policies}\label{sec: other binary policies}

As aforementioned, although the straightforward policy is simple, natural, and practical, it is not an optimal binary policy. An optimal binary policy may perform better, and its performance could be comparable with the ternary policy. Nevertheless, the construction of the optimal policy, which requires a $T$-step look ahead, is computationally intractable for large $T$. We instead present two other binary policies, the myopic policy, and the EvE (exploration versus exploitation) policy, to demonstrate how the performance of the straightforward policy could be improved while maintaining the binary message space. We also discuss the shortcomings of such policies.

The two binary policies considered in this section can be characterized by a threshold parameter, $\rho_t$. Specifically, the policy decides the message according to the following criterion.
\begin{equation}
    \mu_t(z_t; \rho_t) = 
    \left\{
    \begin{array}{ll} 
    1 & \text{ if } z_t > \rho_t; \\
    -1 & \text{ otherwise}.
    \end{array}
    \right.
\end{equation}
Note that the straightforward policy also belongs to this policy class, where the threshold parameter is fixed to $\rho_t^{\mathrm{st}} \coloneqq - x_t m_t$ for all $t$.

\subsection{Myopic Policy}\label{subsec: exploitative policy}

\subsubsection{Definition}

We first analyze whether and how the recommender can improve the current user's payoff by fully exploiting the recommender's current information. For simplicity, we focus on the case of $x_t > 0$. The analysis for the case of $x_t < 0$ is similar, while we need to flip some inequalities appearing in the calculation process. Let $V(\rho_t; x_t, l_t, u_t)$ be the current user's expected payoff, where the expectation is taken with respect to the recommender's current information, $(x_t, l_t, u_t)$:
\begin{align}
    V(\rho_t; x_t, u_t, l_t) & = \Ex_{\tiltheta \sim \Unif[l_t, u_t], z_t \sim \Normal}\left[
    \Ind\{b_t = 1\}\left(x_t \tilde{\theta} + z_t\right)\right] \\
    & = \frac{1}{u_t - l_t}\left[
    \begin{gathered}
    \int_{\min\left\{\max\left\{-\frac{\Ex[z'_t | z'_t < \rho_t]}{x_t}, l_t\right\}, u_t\right\}}^{u_t}\int_{-\infty}^{\rho_t} \left(x_t \theta + z_t \right)\phi(z_t)  dz_t d\theta\\
    +\int_{\min\left\{\max\left\{-\frac{\Ex[z'_t | z'_t > \rho_t]}{x_t}, l_t\right\}, u_t\right\}}^{u_t}\int_{\rho_t}^{\infty} \left(x_t \theta + z_t \right)\phi(z_t)  dz_t d\theta
    \end{gathered}
    \right].\label{eq: myopic user payoff}
\end{align}
According to the recommender's (Bayesian) posterior belief, the state $\theta$ is distributed according to $\Unif[l_t, u_t]$. If the state $\theta$ is so large that the user's expected payoff from arm $1$ is larger than zero, then the user chooses arm $1$ and receives a payoff of $x_t \theta + z_t$. Otherwise, the user chooses arm $-1$ and receives a zero payoff, which does not appear in the formula \eqref{eq: myopic user payoff}. When $a_t = 1$ is recommended, the user knows that $z_t > \rho_t$, and the user chooses arm $1$ if and only if
\begin{equation}
    x_t \theta + \mathbb{E}[z'_t| z'_t > \rho_t] > 0,
\end{equation}
or equivalently,
\begin{equation}
    \theta > - \frac{\mathbb{E}[z'_t| z'_t > \rho_t]}{x_t}.
\end{equation}
Similarly, when $a_t = -1$ is recommended, the user chooses arm $1$ if and only if
\begin{equation}
    \theta > - \frac{\mathbb{E}[z'_t| z'_t < \rho_t]}{x_t}.
\end{equation}
The minimum and maximum appearing in the interval of integration is for letting the threshold within the belief support, $[l_t, u_t]$. The formula \eqref{eq: myopic user payoff} is obtained by specifying the region of $\theta$ under which $b_t = 1$ will be taken. The \emph{myopic policy} maximizes $V$, i.e., $\rho^{\mathrm{myopic}}_t \coloneqq \argmax_{\rho_t} V(\rho_t; x_t, l_t, u_t)$.

\subsubsection{Characterization and Regret Rate}

Using direct calculation, we can derive the functional form of $V$.
\begin{empheq}[left = {V(\rho_t) = \empheqlbrace}]{alignat = 2}
    &\frac{1}{2}\frac{1}{u_t - l_t}\left[x_t u_t^2 + \frac{1}{x_t}\frac{(\phi^*)^2}{\Phi^* (1 - \Phi^*)} \right] & \quad & \text{if } x_t l_t < - \frac{\phi^*}{1 - \Phi^*}<\frac{\phi^*}{\Phi^*}< x_t u_t, \tag{Case 1}\\
    &\begin{aligned}
    &x_t m_t (1 - \Phi^*) + \phi^* \\
    &+ \frac{1}{2(u_t - l_t)} \left[x_t u_t^2 \Phi^* + \frac{(\phi^*)^2}{\Phi^* x_t} - 2 \phi^* u_t \right]
    \end{aligned}
    & \quad & \text{if } - \frac{\phi^*}{1 - \Phi^*}< x_t l_t < \frac{\phi^*}{\Phi^*} < x_t u_t,  \tag{Case 2}\\
    & \begin{aligned}
    &x_t m_t (1 - \Phi^*) + \phi^* \\
    & + \frac{1}{2(u_t - l_t)} \left[
    \begin{aligned}
    & x_t l_t^2 (1 - \Phi^*) \\
    & +\frac{(\phi^*)^2}{(1 - \Phi^*) x_t} + 2 \phi^* l_t
    \end{aligned}
    \right]
    \end{aligned}
    & \quad & \text{if }  x_t l_t < -\frac{\phi^*}{1 - \Phi^*} < x_t u_t < \frac{\phi^*}{\Phi^*}, \tag{Case 3}\\
    & x_t m_t (1 - \Phi^*) + \phi^* & \quad & \text{if } -\frac{\phi^*}{1 - \Phi^*} < x_t l_t <  x_t u_t < \frac{\phi^*}{\Phi^*}, \tag{Case 4}\\
    &x_t m_t & \quad & \text{if }  \frac{\phi^*}{\Phi^*} < x_t l_t, \tag{Case 5}\\
    &0 & \quad & \text{if } x_t u_t < - \frac{\phi^*}{1 - \Phi^*}, \tag{Case 6}
\end{empheq}
where $\phi^* = \phi(\rho_t)$ and $\Phi^* = \Phi(\rho_t)$. Note that the above formula is derived by assuming $x_t > 0$, and we have a slightly different formula if $x_t < 0$. We can obtain the optimizer, $\rho_t^{\mathrm{myopic}}$, by numerically maximizing the $V$ function.

Depending on the integration intervals that appear in \eqref{eq: myopic user payoff}, the $V$ function takes different forms. The intervals of integration are determined by comparing the following four terms: (i) $x_t l_t$, (ii) $x_t u_t$, (iii) $-\Ex[z'_t | z'_t > \rho_t] = - \phi^*/(1 - \Phi^*)$, and (iv) $-\Ex[z'_t | z'_t < \rho_t] = \phi^*/\Phi^*$. 
Term (ii) is always larger than term (i), and term (iv) is always larger than term (iii).
Accordingly, there are $4!/(2! \times 2!) = 6$ cases in total.

In Cases~5 and 6, the user always (i.e., for any values of $\theta \in [l_t, u_t]$) chooses arms $1$ and $-1$ respectively, regardless of the recommendation $a_t$. In these cases, the recommendation is totally useless for the user's decision-making, and therefore, such a choice of the threshold parameter $\rho_t$ is always suboptimal.

In Cases~1 and 2, the user may deviate and choose $b_t = -1$ when $a_t = 1$ is recommended. In Cases~1 and 3, the user may deviate and choose $b_t = 1$ when $a_t = -1$ is recommended. In Case~4, the user always follows the recommendation. Depending on $(x_t, l_t, u_t)$, an optimal threshold $\rho_t$ could exist in each of the four cases.

If the optimal solution belongs to either Case~2 or 3, then it cannot be represented in a tractable closed-form formula. By contrast, when the optimal $\rho_t$ belongs to Case~1 or 4, then it takes a simple form. When the optimal solution belongs to Case~1, then $\rho_t = 0$ must be the case, and when the optimal solution belongs to Case~4, then $\rho_t = - x_t m_t$ must be the case. These facts can be easily verified by checking the first-order condition for optimality.

\begin{figure}[t!]
    \centering
    \includegraphics[width = 0.8 \textwidth]{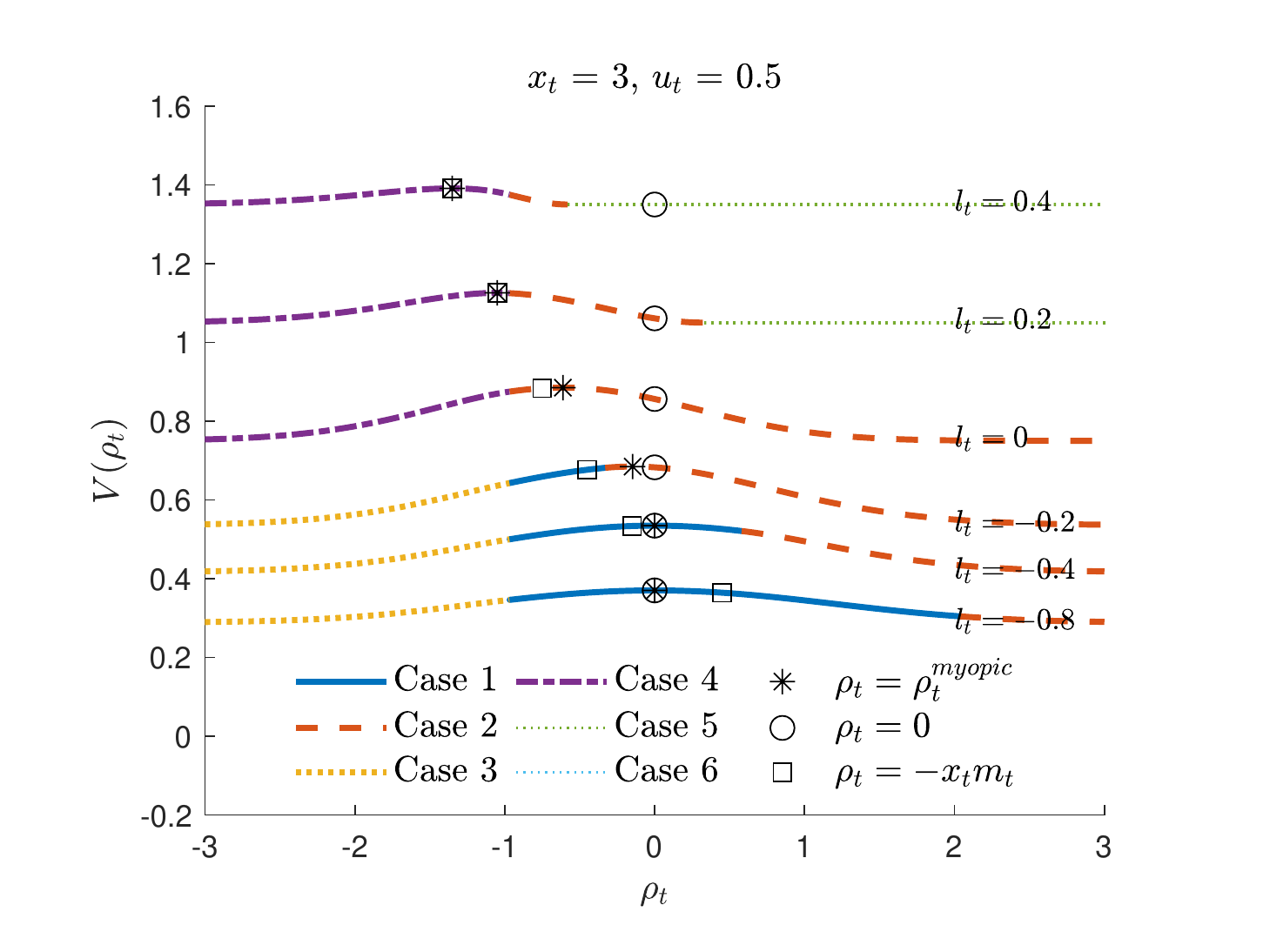}
    \caption{The relationship between the threshold parameter $\rho_t$ and the current user's expected payoff $V(\rho_t)$. We fix $(x_t, u_t)$ to $(3.0, 0.5)$, and plot lines for $l_t \in \{ -0.8, -0.4, 0.2, 0, 0.2, 0.4 \}$. Threshold $\rho$ belonging to different cases are drawn in different colors and styles. The markers $*$, $\circ$, and $\square$ shows the optimal value, the value at $\rho_t = 0$, and the value at $\rho_t = -x_t m_t$, respectively.}
    \label{fig:myopic_optimal}
\end{figure}

Figure~\ref{fig:myopic_optimal} shows the structure of the myopic policy. We fix $(x_t, u_t)$ at $(3.0, 0.5)$, and plot $V$ varying $l_t$. When the confidence interval is wide ($l_t = -0.8$ and $-0.4$), the optimal solution belongs to Case~1, and it is optimal to choose $\rho_t = 0$. In contrast, when the confidence interval is narrow ($l_t = 0.2$ and $0.4)$, the optimal solution belongs to Case~4, and it is optimal to choose $\rho_t = - x_t m_t$. For the intermediate case ($l_t = - 0.2$ and $0$), the myopically optimal policy is also intermediate: The optimal $\rho_t$ belongs to Case~2, and $\rho_t^{\mathrm{myopic}}$ is also in between $\rho_t = 0$ and $\rho_t = - x_t m_t$.

\subsubsection{Connection to the Straightforward Policy}

When the optimal solution belongs to Case~4, the myopic policy matches the straightforward policy. Figure~\ref{fig:myopic_optimal} suggests that this is likely to occur if the confidence interval is small. The following theorem formally demonstrates that when the confidence interval $[l_t, u_t]$ is sufficiently small, these two policies are identical unless $|x_t|$ is very large.

\begin{thm}[Equivalence under Narrow Confidence Intervals]\label{thm_binary_shortopt}
For all $\delta > 0$, for all $(l_t, u_t)$ such that $u_t - l_t \eqqcolon w_t < \delta^2/4 \sqrt{\pi \log(1/\delta)}$, we have
\begin{equation}
    \Prob_{x_t \sim \Normal}\left[\left.\rho^{\mathrm{myopic}}_t = -x_t m_t \right| l_t, u_t\right] \ge 1-\delta.
\end{equation}
\end{thm}

The intuition of Theorem~\ref{thm_binary_shortopt} can be described as follows. When $w_t$ is small and $|x_t|$ is not very large, the recommender's estimation about $r_t(1) = x_t \theta + z_t$ is accurate, and therefore, the recommender can correctly figure out the better arm with a large probability. In such a case, the recommender should communicate which arm is estimated to be better, and the user should follow the recommendation. Therefore, the myopic policy matches the straightforward policy.

Theorem~\ref{thm_binary_shortopt} is important for our argument for two reasons. First, Theorem~\ref{thm_binary_shortopt} characterizes the asymptotic efficiency of the straightforward policy. That is, while the straightforward policy may differ from the myopic policy at first, the two policies align in the long run, indicating that the straightforward policy asymptotically maximizes the current user's payoff. In this sense, the straightforward policy is asymptotically optimal in terms of exploitation. Second, given that the myopic policy and the straightforward policy are asymptotically identical, their regret rates are also identical. Accordingly, the myopic policy also fails to learn the state $\theta$ precisely even in the long run and incurs $\tilTheta(T)$ regret. 

\subsubsection{Drawbacks}

This section’s analyses have demonstrated that the straightforward policy is sometimes suboptimal for the current user. In particular, when the recommender is not confident about $\theta$, the recommender should simply communicate the sign of $z_t$, rather than communicating the estimated better arm. When the recommender has moderate confidence, the optimal threshold falls between the two. By changing the structure of the recommendation depending on the width of the confidence interval, the recommender can improve current users' payoffs.

Nevertheless, the myopic policy is difficult to implement. Unlike the straightforward policy, the interpretation of the message $a_t$ changes over time. For example, in the car navigation context, when the confidence interval is wide, $a_t$ communicates the sign of $z_t$, which is interpreted as ``which road is more vacant''; meanwhile, when the confidence interval is narrow, $a_t$ communicates the sign of $x_t m_t + z_t$, which is interpreted as ``which road the recommender estimates to be better.'' Furthermore, the message's interpretation shifts continuously in the intermediate case.
Therefore, while we have considered binary policies because a binary message space to evaluate the regret achieved by minimal communication, the myopic policy requires complex communication. We contend that the ternary policy should be easier for users than the myopic policy in practice.

Moreover, Section~\ref{sec_sim} shows that the myopic policy and the straightforward policy perform very similarly. More importantly, the ternary policy greatly outperforms these two.

\subsection{Exploration versus Exploitation (EvE) Policy}\label{subsec: exploratory policy}

\subsubsection{Definition}

Parallel to the multi-armed bandit problem, the recommender faces the tradeoff between the acquisition of new information (called ``exploration'') and optimization of her decision based on current information (called ``exploitation''). The straightforward policy and myopic policy only consider the current user's payoff, and therefore, tend towards exploitation.

In this section, we study the EvE policy. This policy initially explores the information on  state $\theta$, ignoring the current user's expected payoff. That is, when the width of the confidence interval, $w_t$, is larger than a threshold, $1/\sqrt{T}$, the recommender attempts to maximize the recommender's own information gain. Subsequently, the recommender starts to exploit the acquired information, i.e., adopts the straightforward policy.\footnote{Since the confidence interval is already narrow, the straightforward policy is approximately optimal in terms of exploitation (Theorem~\ref{thm_binary_shortopt}).} Since the per-round regret from the straightforward policy is $\Theta(w_t^2) = \Theta(1/T)$, the total regret from the exploitation phase is bounded by a constant. Accordingly, the total regret rate is characterized by the learning speed in the exploration phase. Such policies often outperform myopic policies in multi-armed bandit problems, where decision makers need to consider the tradeoff between exploration and exploitation.

The formal construction of the exploration phase is as follows. We design a policy such that user $t$'s action $b_t$ discloses whether $\theta > m_t$ or not. To this end, we define the threshold function $c: \Real \to \Real$ as follows. The value of $c(0)$ is defined arbitrarily.\footnote{We need to define $c(0)$ because $m_t = 0$ in the first round. Once the confidence interval is updated, $x_t m_t = 0$ subsequently occurs with probability zero, meaning the definition of $c(0)$ does not impact the long-run performance of the EvE policy.} For $y > 0$, $c(y)$ is defined as a unique scalar that satisfies
\begin{equation}\label{eq: threshold function c(y) definition}
    \left(\Ex[z_t | z_t < c(y)] =\right) \quad -\frac{\phi(c(y))}{\Phi(c(y))} = - y.
\end{equation}
Such $c(y)$ exists because as $c(y)$ moves from $-\infty$ to $\infty$, the left hand side of \eqref{eq: threshold function c(y) definition} moves from $-\infty$ to $0$. Furthermore, the solution is unique because the left hand side of $\eqref{eq: threshold function c(y) definition}$ is increasing in $c(y)$. For $y < 0$, $c(y)$ is defined as a unique scalar that satisfies
    \begin{equation}
        \left(\Ex[z_t | z_t > c(y)] = \right) \quad \frac{\phi(c(y))}{1 - \Phi(c(y))} = -y.
    \end{equation}
The existence and uniqueness of $c(y)$ for $y < 0$ can be shown in the same manner. The shape of the threshold function $c$ appears in Figure~\ref{fig: cfunction}.

\begin{figure}[t!]
    \centering
    \includegraphics[width = 0.5 \textwidth]{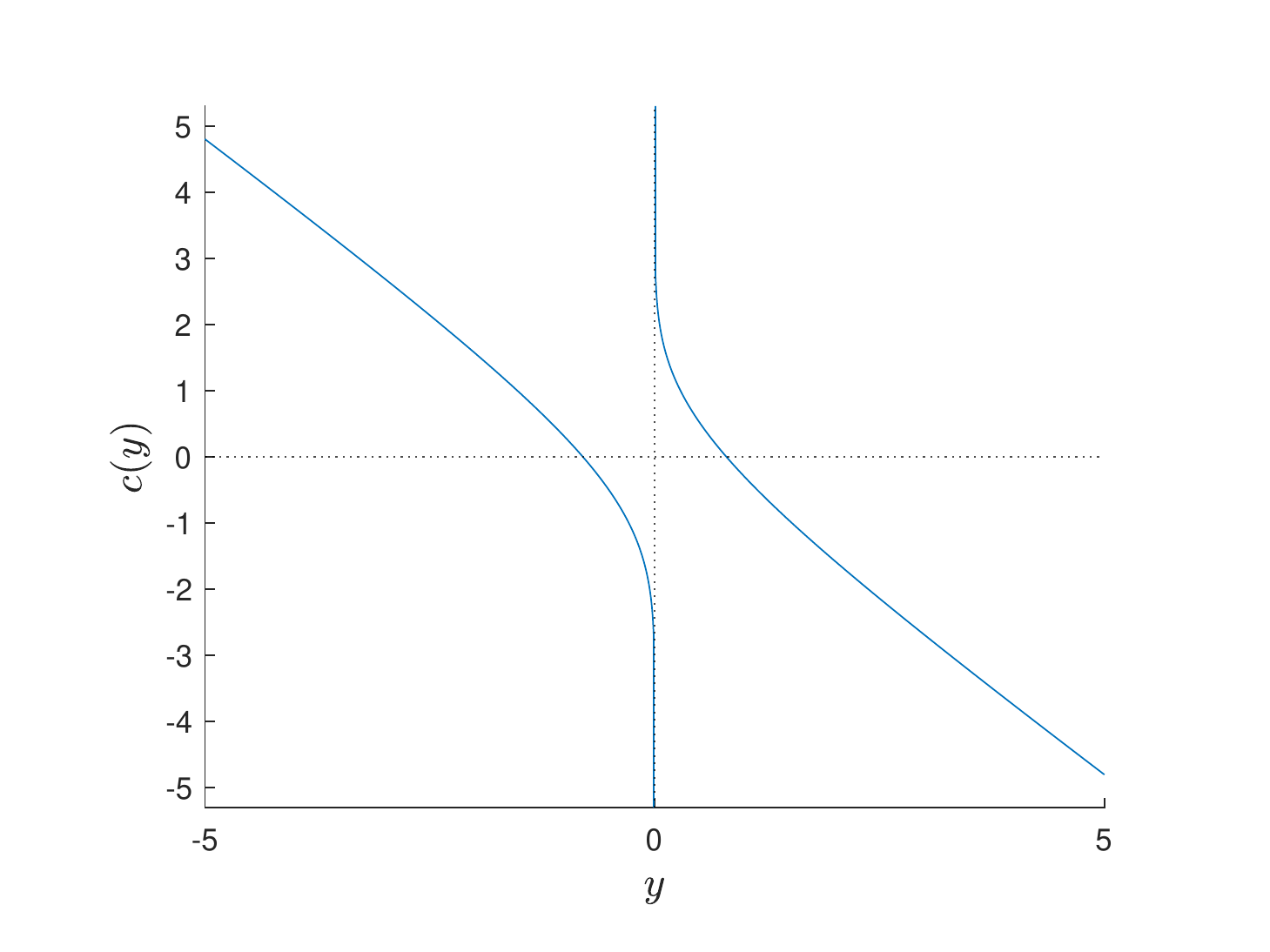}
    \caption{The shape of the threshold function $c$. $\lim_{y \uparrow 0} c(y) = -\infty$, $\lim_{y \downarrow 0} c(y) = + \infty$, and the value of $c(0)$ is defined arbitrarily.}
    \label{fig: cfunction}
\end{figure}

We define the EvE policy as follows:
\begin{equation}
    \rho_t^{\mathrm{EvE}} \coloneqq \begin{dcases}
    c(x_t m_t) & \text{ if } w_t > \frac{1}{\sqrt{T}} \hspace{1em} \text{ (Exploration Phase)},\\
    -x_t m_t & \text{ otherwise.} \hspace{1em} \text{ (Exploitation Phase)}.
    \end{dcases}
\end{equation}

\subsubsection{Regret Rate}

For instruction, we focus on the case of $x_t m_t > 0$.
Suppose that the recommender employs $\rho_t = c(x_t m_t)$.
If $z_t < c(x_t m_t)$ is the case, message $a_t = -1$ is sent, and then user $t$'s expected payoff from arm $1$ is $r_t(1) = x_t \theta + \Ex[z'_t | z'_t < c(x_t m_t)] = x_t(\theta - m_t)$. Since $x_t$ is known, by observing the user's choice, the recommender can identify whether $\theta > m_t$ or not. By observing the user's choice $b_t$, the recommender can halve the confidence interval: From $[l_t, u_t]$ to either $[l_t, m_t]$ or $[m_t, u_t]$. Furthermore, $z_t < c(x_t m_t)$ occurs with a constant probability for each round. Accordingly, under the threshold policy with $\rho_t = c(x_t m_t)$, the confidence interval shrinks geometrically. Therefore, it takes only $O(\log T)$ rounds to reach $w_t < 1/\sqrt{T}$.

The following theorem demonstrates that the EvE policy achieves $O(\log T)$ expected regret.
\begin{thm}[Regret Bound of EvE]
\label{thm: EvE regret bound}
Under the EvE policy, there exists a universal constant $C_{\mathrm{EvE}} > 0$ such that the regret is bounded as:
\begin{equation}
    \Ex[\Regret(T)] \le C_{\mathrm{EvE}} \log T.
\end{equation}
\end{thm}

The EvE policy spends the first $O(\log T)$ rounds for learning and subsequently adopts the straightforward policy. It incurs $O(\log T)$ regret for the exploration phase, and the regret from the exploitation phase is bounded by a constant. The EvE policy outperforms the straightforward policy in terms of the regret rate.

\begin{remark}
The construction of the exploration phase of the EvE policy resembles the optimal information design in a Bayesian persuasion model \citep{KG2011}. \cite{KG2011} consider a sender-receiver model, where the sender submits information about the state, and the receiver takes an action to maximize his payoff that depends on the state and action. \cite{KG2011} show that the sender can maximize her own payoff by obscuring the state information conveyed to the receiver, and in some cases, the optimal information design makes the receiver indifferent between multiple actions. In the exploration phase of the EvE policy, the recommender (sender) also attempts to make the user (receiver) indifferent between the two arms with respect to the recommender's best knowledge (i.e., believing $\theta = m_t$), by setting $\rho_t = c(x_t m_t)$. However, its purpose differs, with the recommender in our model is attempting to extract more information from the user rather than induce a recommender-preferred action.
\end{remark}

\subsubsection{Drawbacks}

Although the EvE policy achieves a sublinear regret rate without expanding the message space, several drawbacks limit its practical desirability.

First, the EvE policy sacrifices the utility of early users. This feature produces an unfair welfare distribution across users. Furthermore, although this paper does not model the outside option, when users find that the recommender is not really helping them but trying to take advantage of their knowledge, they may quit using the recommender system. If this is the case, the recommender fails to extract users' knowledge in the exploration phase.

Second, the EvE policy needs detailed information about the environment. To compute the threshold function $c$, the recommender needs detailed knowledge about the distribution of $z_t$. Furthermore, to optimize the length of the exploration phase, the recommender should know the total number of users, $T$. The EvE policy cannot be used if the recommender has no such detailed knowledge. By contrast, the straightforward policy and ternary policy can be implemented even in the absence of such information.

Third, although the EvE policy outperforms other binary policies, allowing a ternary message space makes achieving an even better regret rate easy.
Recall that the simple ternary policy considered in Section~\ref{sec_ternary} has a regret rate of $O(1)$, even though its construction does not fully account for the tradeoff between exploration and exploitation. Given that the ternary policy does not suffer from the disadvantages associated with the EvE policy, we have no reason to implement such a policy. In Section~\ref{sec_sim}, we further demonstrate that the ternary policy substantially outperforms the EvE policy.

\section{Simulations}\label{sec_sim}

This section provides the simulation results. 
For each path, we draw $\theta$ from $\Unif[-1, 1]$ and $x_t, z_t$ from i.i.d. standard normal distribution for each of the $T = 10,000$ rounds. For the ternary policy, we choose $\Ceps = 1/4$ as the algorithm parameter.\footnote{The simulation code is available at \url{https://github.com/jkomiyama/deviationbasedlearning}.}

\subsection{Regret Growth}

\begin{figure}[t!]
    \centering
    \begin{minipage}[t]{0.48\textwidth}
         \centering
         \includegraphics[width=\textwidth]{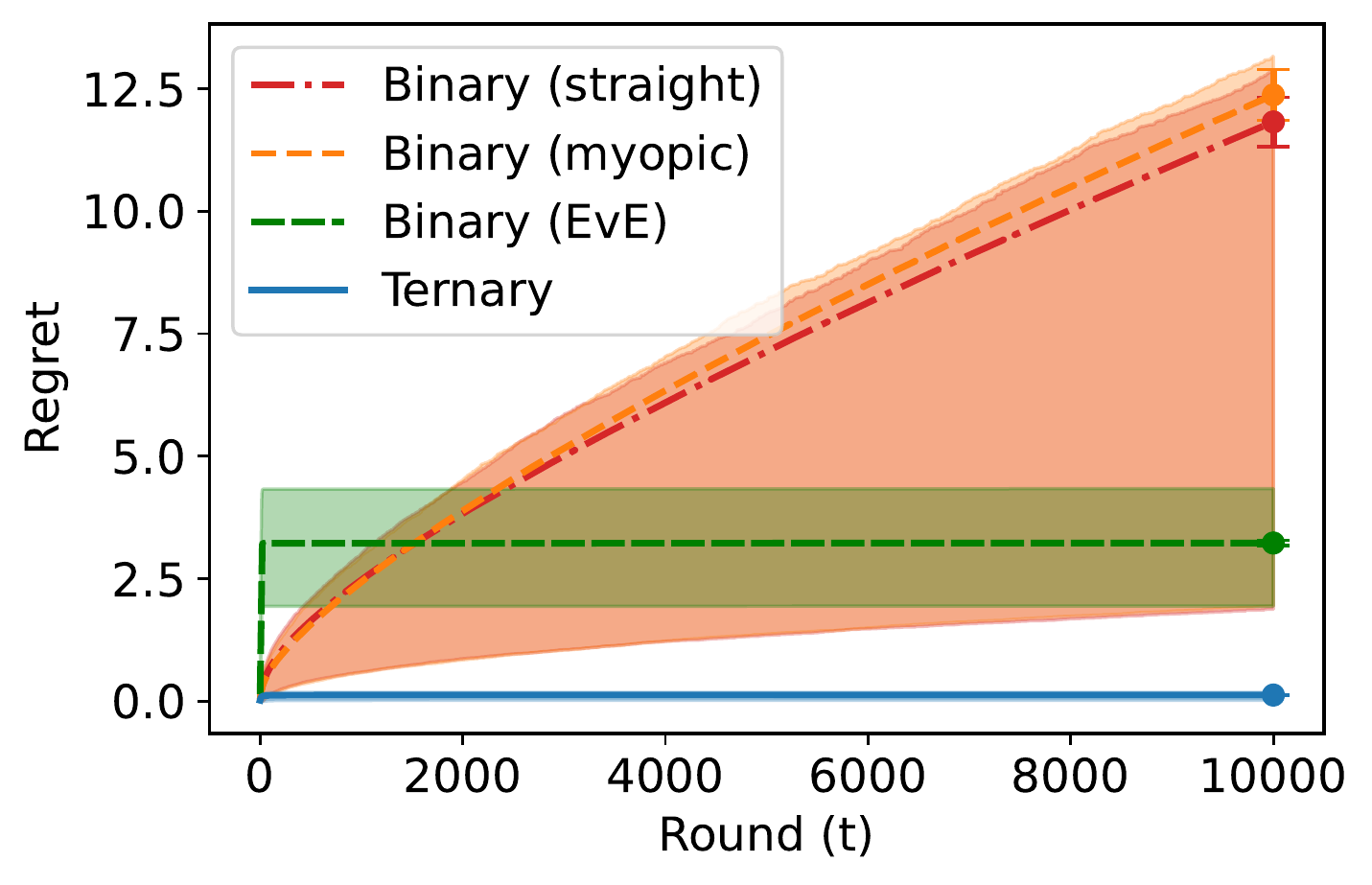}
        (a) Comparison of the four policies
    \end{minipage}
    \hspace{0.02\textwidth}
    \begin{minipage}[t]{0.48\textwidth}
         \centering
         \includegraphics[width=\textwidth]{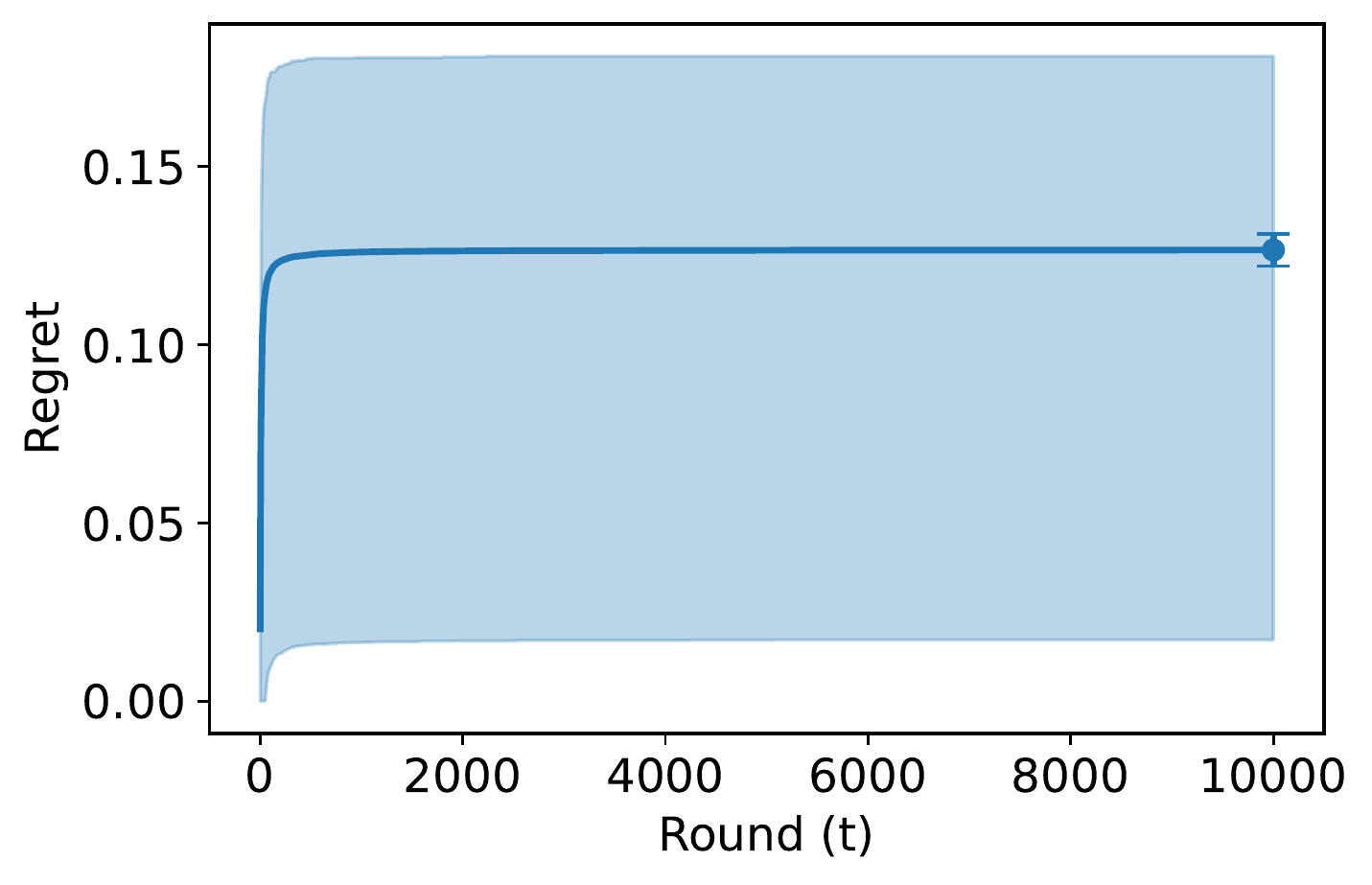}
        (b) Enlarged figure
    \end{minipage}
    \caption{The evolution of cumulative regret $\Regret(t)$ under the straightforward policy, myopic policy, EvE policy, and ternary policy. Panel (b) is an enlarged view of Panel (a).}
    \label{fig:regret}
\end{figure}

We plot the cumulative regret, $\Regret(t)$, in Figure~\ref{fig:regret}. For all graphs in this section, the lines are averages over $\Runnum$ trials, and the shaded areas cover between $25$ and $75$ percentiles. The whiskers drawn at the final round ($T = 10,000$) represent the two-sigma confidence intervals of the average values.

As Theorem~\ref{thm_binary} has proven, under the straightforward policy, the cumulative regret grows almost linearly, and users suffer from a large regret in the long run. The myopic policy behaves similarly to (but very slightly worse than) the straightforward policy. By contrast, the EvE policy initially explores the value of $\theta$ by sacrificing early users and produces approximately no regret subsequently. Consequently, in the last period ($T = 10,000$), the EvE policy performs substantially better than the other two binary policies: The EvE policy incurs regret of 3.23 on average, whereas the straightforward policy and myopic policy incur 11.82 and 12.37.  Furthermore, if the problem continued beyond round $10,000$, we could expect the performance difference to become larger and larger.

Despite the ternary policy's simple construction, it performs remarkably better than these three binary policies.
Similar to the EvE policy, regret grows rapidly in the beginning. However, regret growth terminates much earlier than under the EvE policy, and subsequently, cumulative regret does not grow. Because its regret is proven to be bounded by a constant (Theorem~\ref{thm_three}), it is guaranteed that regret would not grow even when $T$ is extremely large. At $T = 10,000$, the ternary policy only incurs regret of 0.127, which is roughly $1/100$ of the regret incurred by the straightforward policy.

Interestingly, the performance difference between the ternary and EvE policies is large, despite the EvE policy learning the state at an exponential rate. This is because the EvE policy does not optimally select the ``timing to learn.'' 
The EvE policy sacrifices payoffs of all users arriving during the exploration phase, completely ignoring their situations. That is, each user is not informed of the better arm even when one arm appears much better than the other for him (i.e., $|x_t m_t + z_t | \gg 0$). Such users suffer from large per-round regret, meaning regret grows very rapidly during the exploration phase. By contrast, the ternary policy attempts to extract information from the user only when the current user is estimated to be indifferent between two arms (i.e., $|x_t m_t + z_t| \approx 0$), making the per-round regret much smaller. While the ternary policy learns more slowly than the EvE policy, this feature does not deteriorate the regret. If the two arms exhibit different performances for the current user, the recommender need not collect the information at that moment. That is, although the recommender's knowledge is not yet precise, the recommender can confidently recommend the better arm for such an ``easy case.'' The ternary policy attempts to acquire information when (and only when) the information is necessary for identifying the better arm. Because the EvE policy ignores this aspect, it is inefficient.

\subsection{Exploration}\label{subsec: sim long term}

\begin{figure}[t!]
    \centering
    \includegraphics[width = 0.6 \textwidth]{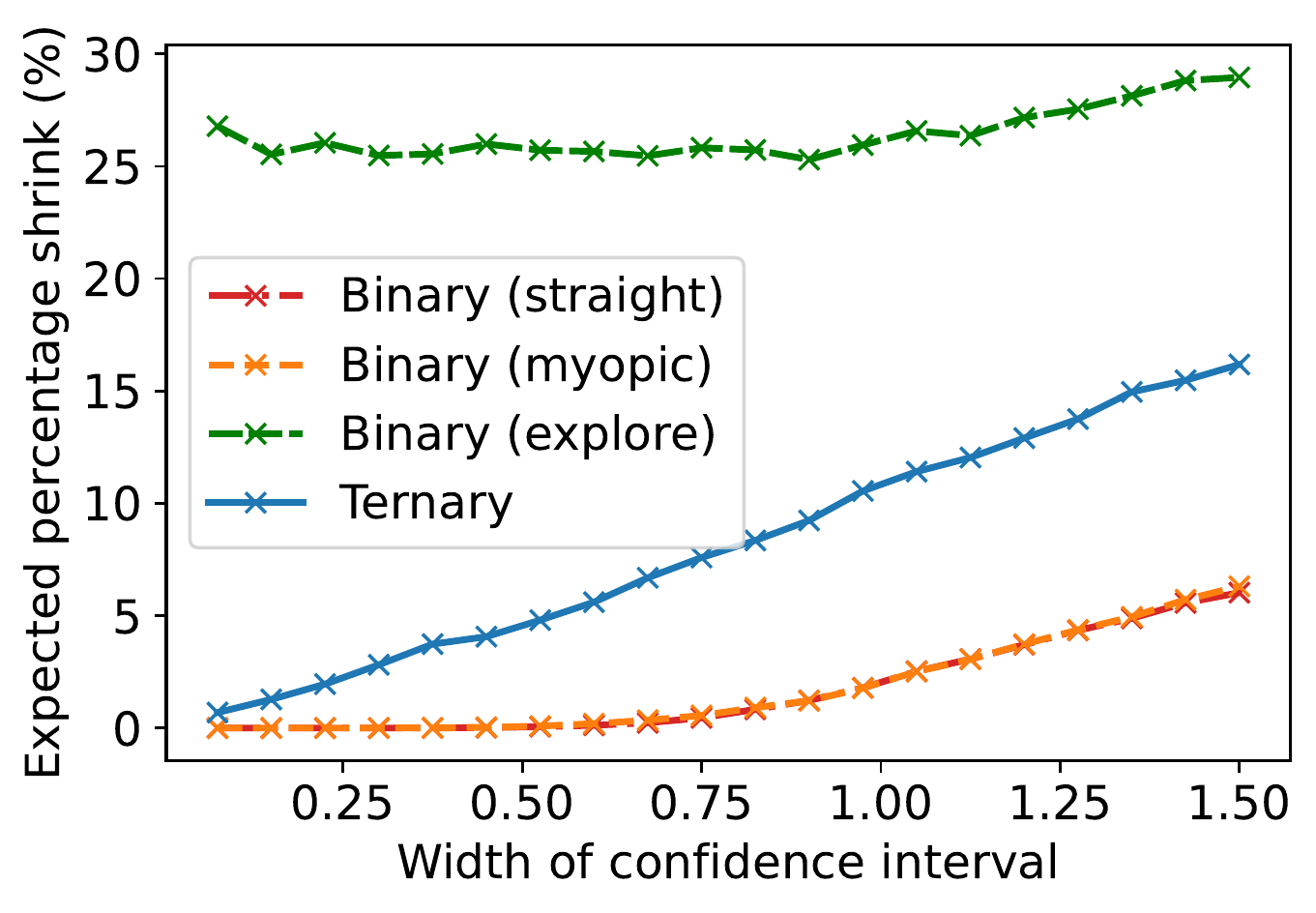}
    \caption{The per-round expected percentage shrink, $\Ex_{\theta \sim \Unif[l_t, u_t]}\left[(w_{t} - w_{t+1})/{w_t}\right]$ under the straightforward policy, the myopic policy, the exploration phase of the EvE policy, and the ternary policy. The performances of the straightforward policy and myopic policy are nearly identical. $u_t$ is fixed at $0.5$, and we vary the width of the confidence interval, $w_t \coloneqq u_t - l_t$, between $0.0$ and $1.5$.}
    \label{fig:single_long}
\end{figure}

This section analyzes each policy's performance in terms of exploration and exploitation of information. First, we depict exploration by showing the learning rate. We measure it by the per-round percentage shrink of the confidence interval, defined as $(w_{t} - w_{t+1})/w_{t}$. If the shrink is large, then the policy learns the state $\theta$ rapidly, and the narrowed confidence interval benefits all subsequent users.

Figure~\ref{fig:single_long} shows the per-round expected percentage shrink, defined by
\begin{equation}
    \Ex_{x_t,z_t \sim \Normal, \theta \sim \Unif[l_t, u_t]}\left[\frac{w_{t} - w_{t+1}}{w_t}\right] \times 100.
\end{equation}
The expectation is replaced with an empirical average of $\Nsingle$ trials. $u_t$ is fixed at $0.5$, and we vary the width of the confidence interval, $w_t \coloneqq u_t - l_t$, between $0.0$ and $1.5$. We plot the performances of the straightforward policy, the myopic policy, the exploration phase of the EvE policy, and the ternary policy.

As Theorem~\ref{thm_binary_shortopt} indicates, the straightforward policy and the myopic policy often generate the same threshold $\rho_t$, rendering their performances almost identical. As anticipated from Theorem~\ref{thm_binary}, 
these two policies perform the worst. While the straightforward policy can acquire some information when $w_t$ is large, the expected shrink diminishes rapidly as $w_t$ becomes small (Lemma~\ref{lem_update} implies that its rate is exponential to $- 1/w_t$). When $w_t = 0.75$, the straightforward policy can shrink the confidence interval only by 0.4\%, and when $w_t = 0.3$, the shrink becomes zero, i.e., literally no information was gained in the 10,000 trials.
(Even for large $w_t$, the EvE policy and the ternary policy substantially outperform the straightforward policy.)

By contrast, the ternary policy effectively shrinks the confidence interval for any $w_t$. 
When $w_t = 0.75$, the ternary policy shrinks the confidence interval by 7.6\%, and even when $w_t = 0.075$, the percentage shrink is 0.7\%.
Therefore, as predicted by theory, the expected percentage shrink is linear in $w_t$. That is, with probability $\Theta(w_t)$, the ternary policy sends $a_t = 0$ (Lemma~\ref{lem_probunknown}), and then, the confidence interval shrinks at a constant percentage (Lemma~\ref{lem_unknown_shrink}). 

The EvE policy's exploration phase is specialized for exploration, and therefore, its expected percentage shrink is much larger than the other policies, and it shrinks the confidence interval by 25\% for any value of $w_t$. The constant information gain is achieved because the policy sends a message ($a_t = -1$ if $x_t m_t > 0$ and $a_t = 1$ if $x_t m_t < 0$) for exploration with a constant probability, and the confidence interval halves every time such a message is sent.

\begin{remark}
In Appendix~\ref{sec: learning}, we analyze how the confidence intervals are updated. Under the straightforward policy, most information gain happens when users deviate from recommendations, while such updates are less frequent. By contrast, under the ternary policy, virtually all information gains happen when $a_t = 0$ is sent.
\end{remark}

\subsection{Exploitation}\label{subsec: sim short term}

\begin{figure}[t!]
    \centering
    \begin{minipage}[t]{0.48\textwidth}
         \centering
         \includegraphics[width=\textwidth]{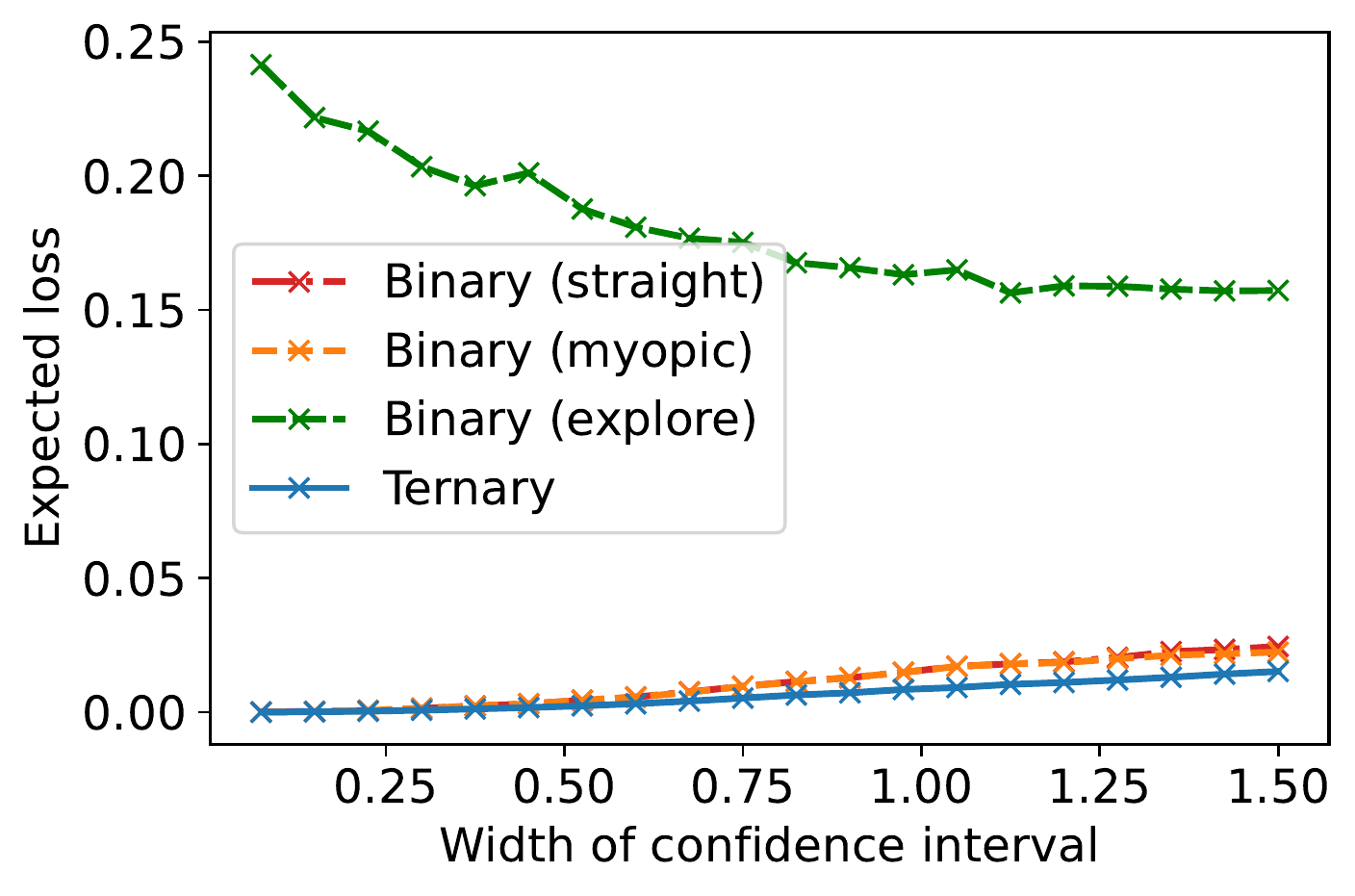}
        (a) Comparison of the four policies
    \end{minipage}
    \hspace{0.02\textwidth}
    \begin{minipage}[t]{0.48\textwidth}
         \centering
         \includegraphics[width=\textwidth]{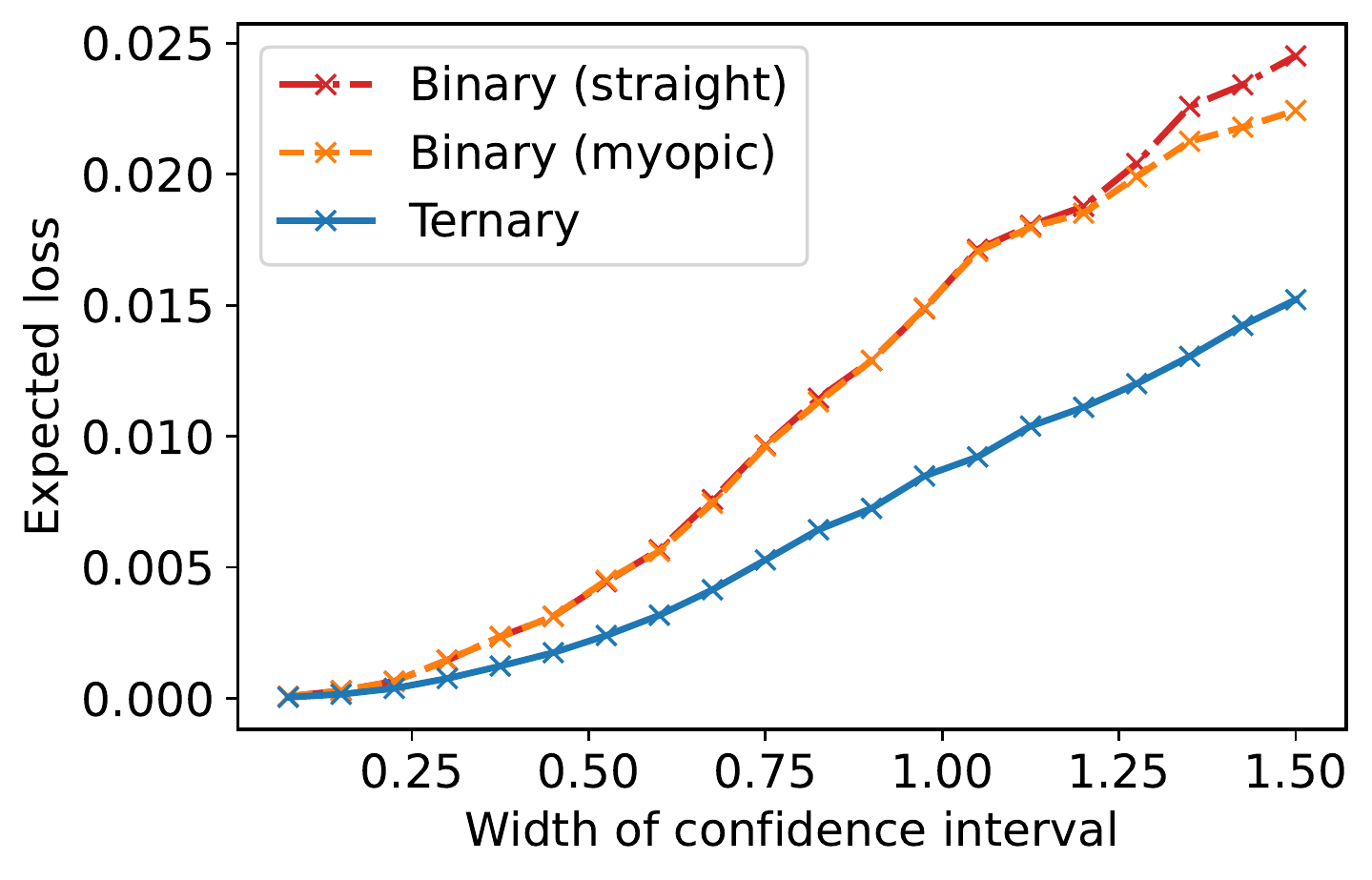}
        (b) Enlarged figure
    \end{minipage}
    \caption{The per-round expected regret $\Ex_{\theta \sim \Unif[l_t, u_t]}[\regret(t)]$ under the straightforward policy, the myopic policy the exploration phase of the EvE policy, and the ternary policy. $u_t$ is fixed at $0.5$, and we vary the width of the confidence interval, $w_t \coloneqq u_t - l_t$, between $0.0$ and $1.5$. Panel (b) is an enlarged view of Panel (a).}
    \label{fig:single_short}
\end{figure}

Next, we compare each policy's ability to exploit the recommender's current knowledge about $\theta$. This ability is measured simply using the per-round regret $\regret(t)$, which represents the size of the current user's payoff. Figure~\ref{fig:single_short} shows the per-round expected regret,
\begin{equation}
    \Ex_{x_t,z_t \sim \Normal, \theta \sim \Unif[l_t, u_t]}\left[\regret(t)\right],
\end{equation}
of the straightforward policy, the exploration phase of the EvE policy, and the ternary policy. As in Section~\ref{subsec: sim long term}, we fix $l_t = - 0.5$ and vary the value of $u_t$ between $-0.5$ and $0.5$. The vertical axis represents $w_t = u_t - l_t$.

Panel (a) shows that the exploration phase of the EvE policy incurs a substantial per-round regret. Because it is not designed to reduce the current user's regret, even when the confidence interval is small, the per-round regret does not diminish. Accordingly, this policy incurs a large regret even if it is only used for a short period of time.

Because the other three policies incur much smaller regret, we show an enlarged view of Panel (a) as Panel (b). The difference between the straightforward policy and the myopic policy is very small, while the myopic policy performs slightly better when the confidence interval is wide.

Panel (b) demonstrates that, under this simulation setting (c.f., $\Ceps = 1/4$), the per-round regret of the ternary policy is about $1/2$ to $2/3$ of the per-round regret of the myopic policy.
Considering that the (cumulative) regret of the ternary policy is $1/100$ of the regret of the myopic policy, the improvement in myopic payoffs is relatively small. This fact articulates that welfare gain from increasing the message space is gained mostly from the improvement in learning rate.

Nevertheless, from a fairness perspective, the improvement in myopic payoffs is very important. The myopic policy maximizes the current user's payoff for every round among all binary policies. Panel (b) demonstrates that, for an appropriate choice of the algorithm parameter, the ternary policy provides even better payoffs for any $w_t$. In this sense, the ternary policy sacrifices no user, in contrast to the EvE policy.

\section{Concluding Remarks}\label{sec: conclusion}

In this paper, we propose deviation-based learning, a novel approach to training recommender systems. In contrast to traditional rating-based learning, we do not assume observability of payoffs (or noisy signals of them). Instead, our approach extracts users' expert knowledge from the data about choices given recommendations. The deviation-based learning is effective when (i) payoffs are unobservable, (ii) many users are knowledgeable experts, and (iii) the recommender can easily identify the set of experts.

Our analysis reveals that the size of the message space is crucial for the efficiency of deviation-based learning. Using a stylized model with two arms, we demonstrated that a binary message space and straightforward policy result in a large welfare loss. After the recommender is trained to some extent, users start to follow her recommendations blindly, and users' decisions are uninformative in terms of advancing the recommender's learning. This effect significantly slows down learning, and the total regret grows almost linearly with the number of users. While we can improve the regret rate by developing more sophisticated binary policies (such as the myopic policy and the EvE policy), a much simpler and more effective solution is to increase the size of the message space. 
Employing a ternary message space allows the recommender to communicate that she predicts that two arms will produce similar payoffs.
User's choices after receiving such a message are extremely useful for the recommender's learning. Thus, making such messages available accelerates learning drastically. Under the ternary policy, total regret is bounded by a constant, and in round $10,000$, the ternary policy only incurs $1/100$ of the regret incurred by the straightforward policy.

While it is not explicitly modeled in this paper, the optimal policy choice should also depend on the magnitude of the communication cost. If the communication cost is extremely large, the recommender would abandon communication and choose $|A| = 1$, and if it is zero, the recommender would choose $|A| = \infty$ to achieve the first-best choices from the beginning. However, our analysis suggests that, for a wide range of ``moderately large'' communication costs, the ternary policy should be an (approximately) optimal choice, because it is extremely more efficient than the binary policies, while the communication cost is (nearly) minimal.

Our analysis of the binary policy suggests one further useful caveat: The recommender should not use the rate at which users follow recommendations as a key performance indicator. When the recommender has an information advantage, a user may follow a recommendation blindly even when the recommendation does not fully incorporate his own information and preference. 
This means that using this performance indicator may inadvertently engender a large welfare loss.

Although we believe that the insight obtained from our stylized model will be useful in general environments (given that the intuitions of our theorems do not rely on the assumptions made for the sake of simplicity), more comprehensive and exhaustive analyses are necessary for practical applications. In practice, observable contexts ($x_t$) are often multi-dimensional. Furthermore, users' payoffs are rarely linear in the parameter ($\theta$), and their functional forms may be unknown ex ante, requiring that the recommender adopt a nonparametric approach. Future studies could investigate deviation-based learning in more complex environments.

\bibliographystyle{ecta}
\bibliography{references.bib}

\pagebreak

\appendix

\part*{Appendix}

\section{Source of Learning}\label{sec: learning}

This section investigates how the width of the confidence interval, $w_t$, is updated. First, we focus on when and how frequently updates occur. We count the occurrence of belief updates, i.e., the number of rounds such that $w_{t+1} < w_t$. Among all rounds in which updates occur, (i) the set of rounds in which the user followed the recommendation is denoted by $\Follow$ (obedience), (ii) the set of rounds in which the user deviated from the recommendation is denoted by $\Deviate$ (deviation), and (iii) the set of rounds in which the recommender did not recommend a particular action is denoted by $\Fence$ (on the fence). Formally, we define
\begin{align}
\Follow(t) &\coloneqq  \{s \in [t]: w_{s+1} < w_s \text{ and } a_s = b_s\}; \\
\Deviate(t) &\coloneqq  \{s \in [t]: w_{s+1} < w_s, a_s\neq 0 \text{ and } a_s \neq b_s\};\\
\Fence(t) & \coloneqq  \{s \in [t]: w_{s+1} < w_s \text{ and } a_s = 0\}.
\end{align}
Note that $|\Fence| = 0$ for the case of the straightforward policy since $a_t = 0$ is never sent.

\begin{figure}[t!]
    \centering
    \begin{minipage}[t]{0.48\textwidth}
         \centering
         \includegraphics[width=\textwidth]{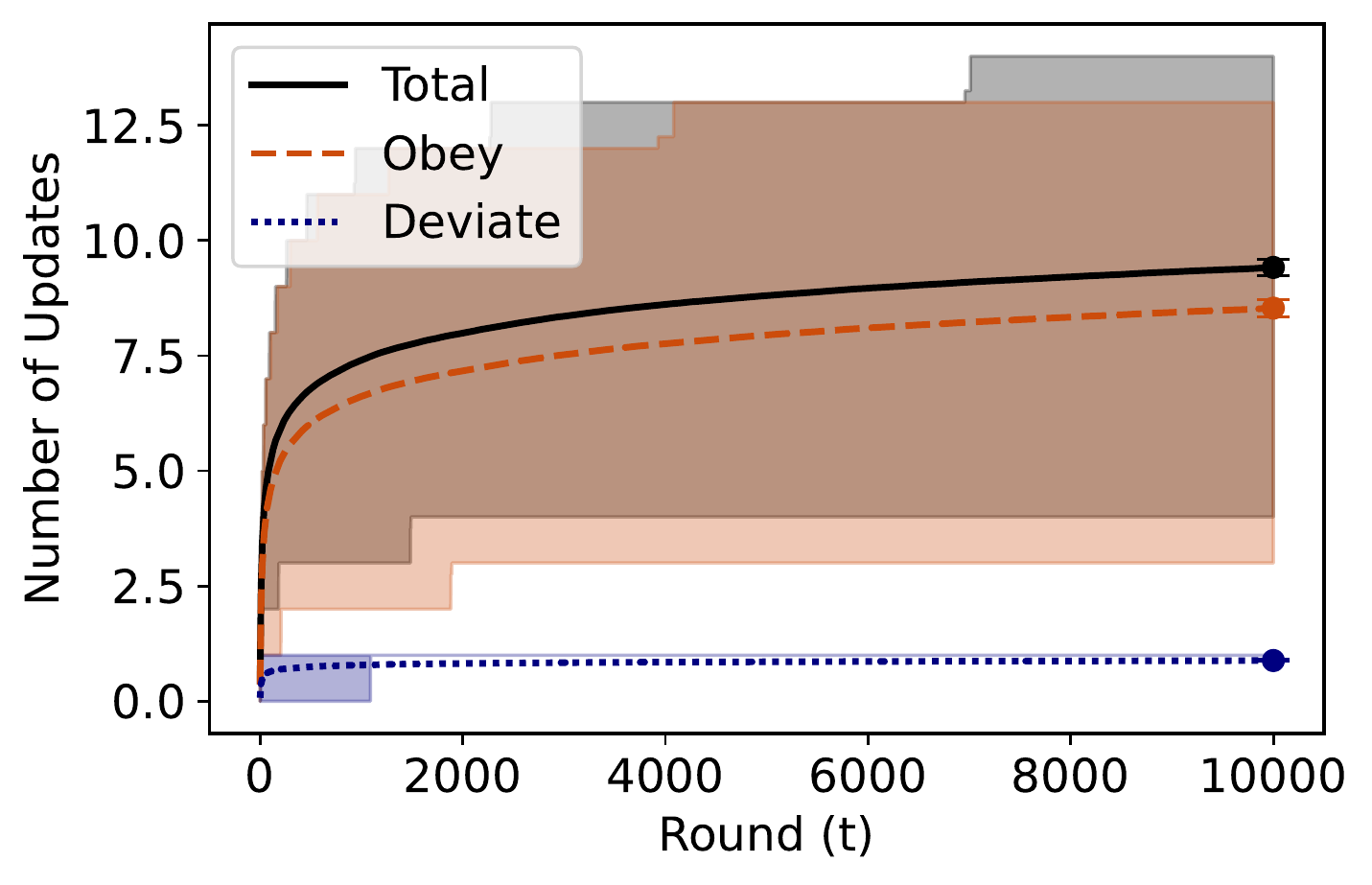}
        (a) Straightforward policy
    \end{minipage}
    \hspace{0.02\textwidth}
    \begin{minipage}[t]{0.48\textwidth}
         \centering
         \includegraphics[width=\textwidth]{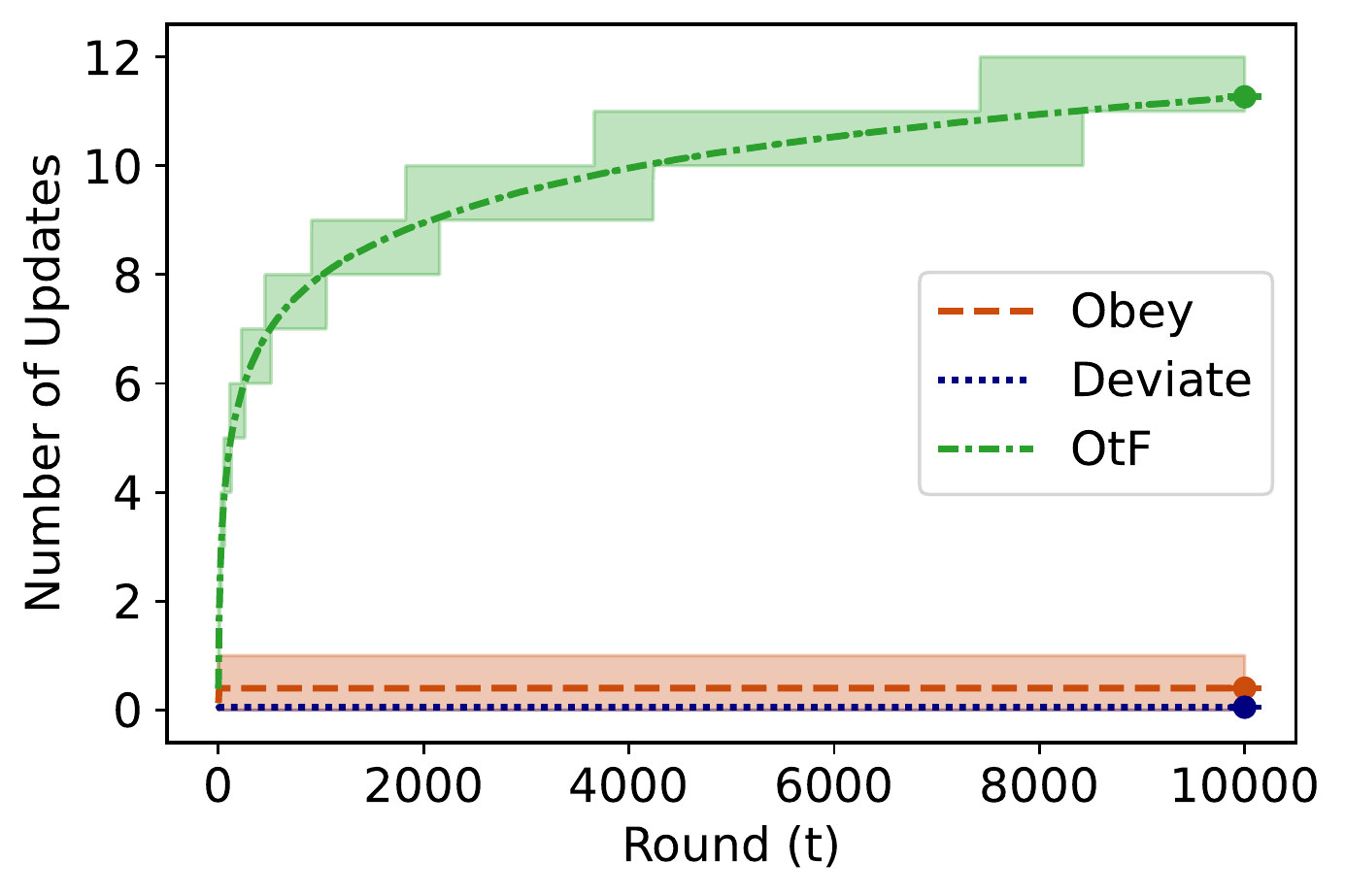}
        (b) Ternary policy
    \end{minipage}
    \caption{The growth of the number of updates. For Panel (b), we observe $|\Follow(T)| = |\Deviate(T)| \approx 0$; thus, the total number of updates until round $t$ is approximately equal to $|\Fence(t)|$.}
    \label{fig: update_num}
\end{figure}

Figure~\ref{fig: update_num} plots the number of updates, $|\Follow(t)|$, $|\Deviate(t)|$, $|\Fence(t)|$, and their total, $|\Follow(t)| + |\Deviate(t)| + |\Fence(t)|$. Panel (a) displays the results of the straightforward policy. The updates by obedience occur more often than the updates by deviation. Panel (b) shows the results of the ternary policy. Since the recommender recommends an arm only if she is confident about it, users follow the recommendation blindly whenever an arm is recommended; therefore, an update occurs only if the recommender confesses that she is on the fence. The opportunities for her learning are mostly concentrated at the beginning, but updates occur occasionally even in the later stages. On average, updates occur more frequently than in the case of the straightforward policy.

Next, we evaluate the total amount of information acquired from each recommendation. We measure the \emph{accuracy} of the estimation in round $t$ by 
\begin{equation}\label{ineq_accuracy}
\ACC(t) \coloneqq -\log \left(w_{t+1}/2\right).
\end{equation} 
The value $w_{t+1}$ is the width of the confidence interval after the round-$t$ update. Note that $w_1 = u_1 - l_1 = 1 - (-1) = 2$, and therefore, $\ACC(0)$ is normalized to zero.

We define the \emph{accuracy gain} from each recommendation as follows:
\begin{align}
\ACCFollow(t) &\coloneqq  -\sum_{s \in \Follow} \log(w_{s+1}/w_{s});\\
\ACCDeviate(t) &\coloneqq  -\sum_{s \in \Deviate} \log(w_{s+1}/w_{s});\\
\ACCFence(t) &\coloneqq -\sum_{s \in \Fence} \log(w_{s+1}/w_{s}).
\end{align}
Note that it is always the case that $\ACC(t) = \ACCFollow(t) + \ACCDeviate(t) + \ACCFence(t)$.

\begin{figure}[t!]
    \centering
    \begin{minipage}[t]{0.48\textwidth}
         \centering
         \includegraphics[width=\textwidth]{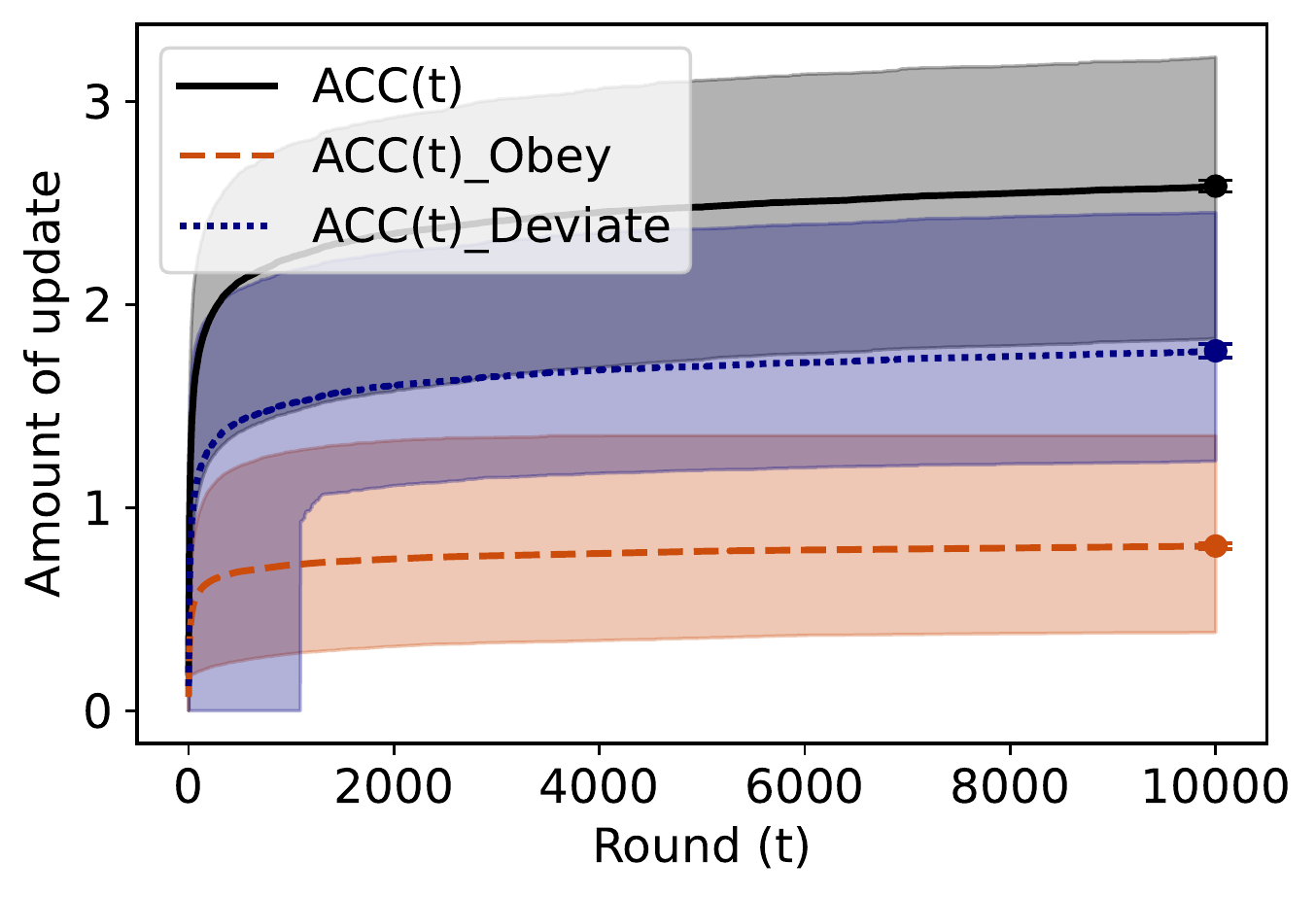}
        (a) Straightforward policy
         \label{fig:update2}
    \end{minipage}
    \hspace{0.02\textwidth}
    \begin{minipage}[t]{0.48\textwidth}
         \centering
         \includegraphics[width=\textwidth]{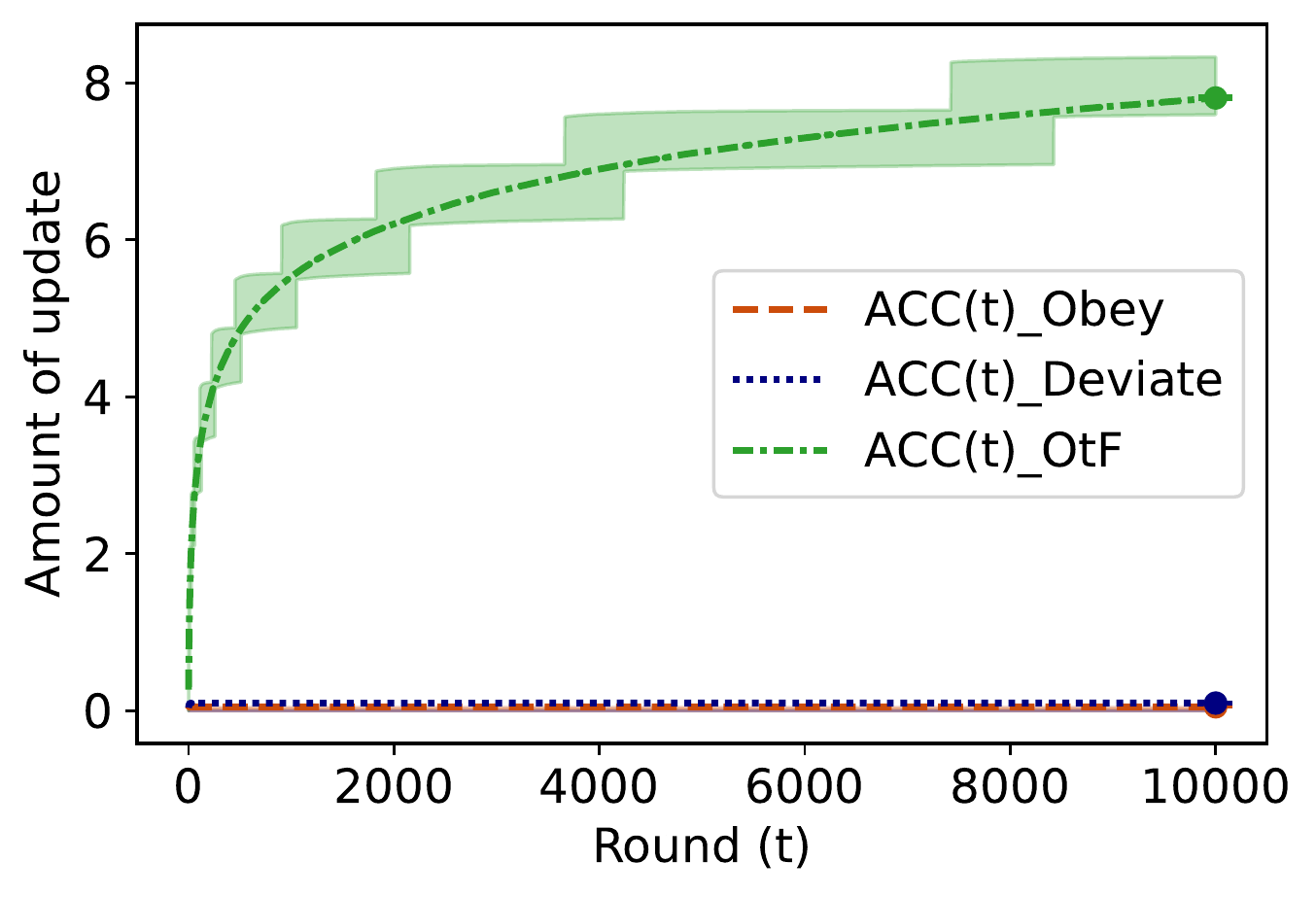}
        (b) Ternary policy
         \label{fig:update3}
    \end{minipage}
    \caption{The breakdown of accuracy gains. For Panel (b), almost all the accuracy gains are from $a_t = 0$; thus, $\ACC(t) \approx \ACCFence(t)$.}
    \label{fig: update}
\end{figure}

Figure~\ref{fig: update} shows the accuracy gain from each recommendation under the straightforward policy and ternary policy. As illustrated in Panel (a), under the ternary policy, learning from obedience occurs more frequently than learning from deviations. Nevertheless, Panel (a) shows that the recommender acquires more information from deviations than obedience. This is because once a deviation occurs, it is much more informative than obedience. 

The following theorem elucidates the informativeness of deviations under the straightforward policy.

\begin{thm}[Informativeness]
\label{lem_deviation}
Under the straightforward policy, if $a_t \ne b_t$, then 
\begin{equation}\label{ineq_wtsize_follow}
w_{t+1} < \frac{1}{2} w_t.
\end{equation}
Conversely, if $a_t = b_t$, then
\begin{equation}\label{ineq_wtsize_oppose}
w_{t+1} > \frac{1}{2} w_t.
\end{equation}
\end{thm}

When a user deviates from the recommendation (i.e., when $a_t \neq b_t$), the width of the recommender's confidence interval will be at least halved. Since the recommender has an information advantage about $z_t$, a deviation occurs only if the recommender significantly misestimates the static payoff component. Accordingly, upon observing a deviation, the recommender can update her belief about $\theta$ by a large amount. In contrast, when a user obeys the recommendation (i.e., when $a_t = b_t$), the decrease of $w_t$ is bounded. Note that, when users are obedient, it is frequently the case that no update occurs and so $w_{t+1} = w_t$. This is the case if the recommender's error $|x_t (\theta - m_t)|$ is small, and therefore, the user follows the recommendation blindly, given any $\theta \in [l_t, u_t]$.

Panel (b) of Figure~\ref{fig: update} reveals that, under the ternary recommendation, almost all the accuracy gains are obtained when the recommender signals are on the fence. For any stage, the learning rate is higher than that under the straightforward policy, and the difference is quantitatively large. In round $10,000$, the average accuracy under the ternary policy becomes larger than that under the straightforward policy by (roughly) six points, which implies that $w_{T}$ under the ternary policy is $e^6 \approx 403$ times smaller than $w_{T}$ under the straightforward policy.

\section{Proofs}\label{sec: proof}

\subsection{Proof of Theorem \ref{thm_binary}}

\begin{proof}[Proof of Theorem \ref{thm_binary}]

Let $C_1 = (1/2)\Cupdate$ and 
\begin{equation}
\EZ(t) := \left\{w_t \le \frac{C_1}{\log T}, |\theta - m_t| \ge \frac{C_1}{2 \log T} \right\}.
\end{equation}
In the following, we first show the following inequality. 
\begin{claim}
\begin{equation}
\Prob\left[\EZ(3)\right] \ge \frac{C_2}{\polylog(T)} \label{ineq_ezprob} 
\end{equation}
for some constant $C_2>0$.
\end{claim}

\begin{proof}

Let
\begin{align} 
\EA &:= \left\{ u_2 \le \theta + \frac{C_1}{6\log T} \right\},\\
\EB &:= \left\{ \theta - \frac{5C_1}{6\log T} \le l_3 \le \theta - \frac{2C_1}{3\log T} \right\}.
\end{align}
Note that $\EA \cap \EB \subseteq \EZ(3)$. 

\begin{figure}
	\centering
	\includegraphics[width = 0.6\textwidth]{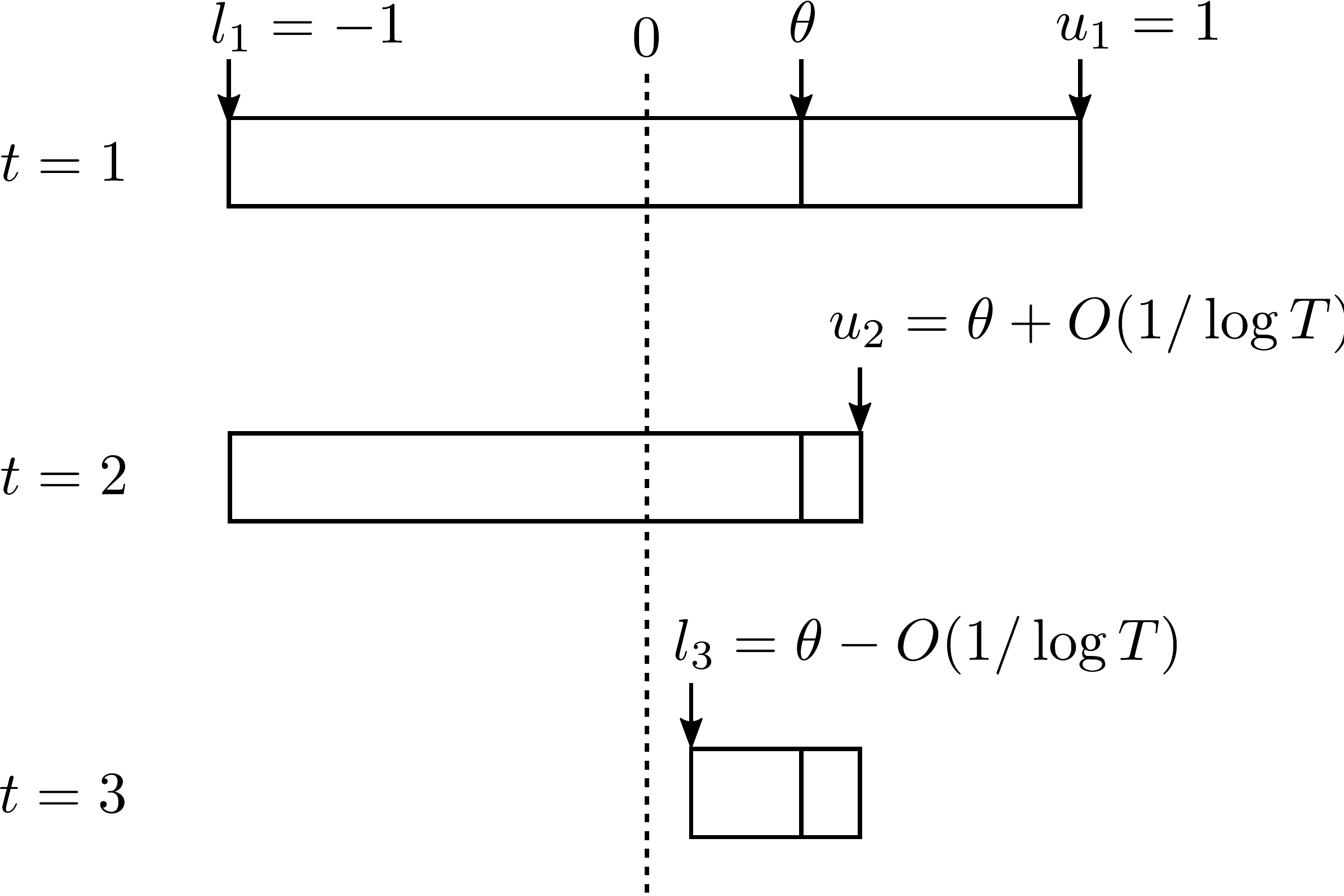}
	\caption{Illustration of $(l_t, u_t)_{t=1,2,3}$ in the instance of Theorem \ref{thm_binary}. Here, $w_3 = O(1/\log T)$ holds, which implies that $(l_t, u_t)$ has very small chance of being updated again.}
	\label{fig:ez_recom_reglower}
\end{figure}

In order to evaluate the probability that $\EZ(3)$ occurs, in the following, we evaluate the probability that $\EA$ and $\EB$ occur, assuming $\theta > 2C_1/(\log T)$ (which occurs with probability $\Theta(1)$ for sufficiently large $T$).

\begin{claim}
$\Prob[\EA] =  \Theta\left(1/(\log T)\right) $.
\end{claim}

\begin{proof}
Recall that $(l_1,u_1,m_1)=(-1,1,0)$. 
Let $\sigma_1 = \int_0^\infty 2 \phi(x) dx$, which is equal to $-\Ehatz{1}$ given $b_1 = -1$. 
Under $a_1 = -1$, $\Ehatz{1} = -\sigma_1$. 
Let
\begin{equation} 
\EA' := \{z_1 < 0\} \cap \left\{\frac{\sigma_1}{\theta + \frac{C_1}{6 \log T}} \le x_1 \le \frac{\sigma_1}{\theta } \right\}.
\end{equation}
In the following, we show that $\EA'$ implies $\EA$ and $\Prob[\EA'] = \Omega(1/\log T)$.
\begin{align}
\EA' 
&= \EA' \cap \left\{a_1 = -1\right\} 
\text{\ \ \ (by $x_1 m_1 + z_1 = z_1 < 0$)} \\
&= \EA' \cap \left\{a_1 = -1, b_1 = -1\right\} 
\text{\ \ \ (by $x_1 \theta + \Ehatz{1} = x_1 \theta - \sigma_1 < 0$)} \\ 
&= \EA' \cap \left\{a_1 = -1, b_1 = -1, u_2 \le \theta + \frac{C_1}{6\log T}\right\} 
\text{\ \ \ (by Eq.~\eqref{ineq_update_right})} \\ 
&\subseteq \EA. \label{ineq_aaprime}
\end{align}

Therefore,
\begin{align}
\Prob[\EA] 
&\ge \Prob[\EA'] \text{\ \ \ (by Eq.~\eqref{ineq_aaprime})}\\
&= \Prob[ z_1 < 0 ] \times
\Prob\left[ \frac{\sigma_1}{\theta + \frac{C_1}{6 \log T}} \le x_1 \le \frac{\sigma_1}{\theta} \right] \\
&= \frac{1}{2} \Prob\left[ \frac{\sigma_1}{\theta + \frac{C_1}{6 \log T}} \le x_1 \le \frac{\sigma_1}{\theta} \right] \\
&= \frac{1}{2} \int_{ 
\frac{\sigma_1}{\theta + \frac{C_1}{6 \log T}} 
}^{
\frac{\sigma_1}{\theta}
} \phi(x) dx \\
&\ge \frac{\sigma_1}{2\theta^2} \frac{C_1}{6\log T}  \times \frac{1}{\sqrt{2\pi}}\exp\left( - \frac{\sigma_1^2}{2\theta^2} \right) 
\text{\ \ \ (for $\theta \ge C_1/(6\log T)$)}
\\
&= \Theta\left(\frac{1}{\log T}\right) \text{\ \ \ (since $\sigma_1, C_1, \theta = \Theta(1)$)}.
\end{align}
\end{proof}

\begin{claim}
The probability $\Prob[\EB|\EA] = \Theta\left(1/(\log T)\right)$.
\end{claim}

\begin{proof}
Let
\begin{equation}
\EB' = \left\{ x_2 > 0 \right\} \cap \left\{ x_2 m_2 + z_2 < 0 \right\} \cap \left\{ \frac{-\Ehatz{2}}{\theta - \frac{2C_1}{3\log T} } \le x_2 \le  \frac{-\Ehatz{2}}{\theta - \frac{5C_1}{6\log T} } \right\}.
\end{equation}
We have
\begin{align}
\EB' 
&= \EB' \cap \left\{a_2 = -1\right\} 
\text{\ \ \ (by $x_2 m_2 + z_2 < 0$)} \\
&= \EB' \cap \left\{a_2 = -1, b_2 = 1\right\} 
\text{\ \ \ (by $x_2 \theta + \Ehatz{2} > 2C_1/(3 \log T) > 0$)} \\
&= \EB' \cap \left\{a_2 = -1, b_2 = 1, \theta - \frac{2C_1}{3\log T} \le l_3 \le \theta - \frac{5C_1}{6\log T}\right\} 
\text{\ \ \ (by Eq.~\eqref{ineq_update_left})} \\
&\subseteq \EB.
\end{align}
Furthermore, by using the fact that $-\Ehatz{2} \in (0, \sigma_1) = \Theta(1)$ and $\sigma_1 = \Theta(1)$, we have the following under $\{l_2 <0, x_2 m_2 < 0\}$:
\begin{align} 
\Prob[\EB|\EA] 
&\ge \Prob[\EB'|\EA] \\ 
&= \Prob\left[
\frac{-\Ehatz{2}}{\theta - \frac{2C_1}{3\log T} } \le x_2 \le  \frac{-\Ehatz{2}}{\theta - \frac{5C_1}{6\log T} },
x_2 m_2 + z_2 < 0
\right]\\ 
&\ge \Prob\left[
\frac{-\Ehatz{2}}{\theta - \frac{2C_1}{3\log T} } \le x_2 \le  \frac{-\Ehatz{2}}{\theta - \frac{5C_1}{6\log T} }
\right] \times \frac{1}{2} \\ 
&\text{\ \ \ \ \ (by $x_2 m_2 \le 0$)} \\
&\ge \frac{C_1}{6\log T} \frac{-\Ehatz{2}}{2\theta^2} \phi\left(\frac{-2\Ehatz{2}}{\theta}\right) \times \frac{1}{2} \\ 
&\text{\ \ \ \ \ (for $\theta \ge 2 \times \frac{5C_1}{6\log T}$)} \\
&= \Theta\left(\frac{1}{\log T}\right). \text{\ \ \ (since $\Ehatz{2}, C_1, \theta = \Theta(1)$)}
\end{align}
\end{proof}

Combining these claims, we have
\begin{equation}\label{eq: ez eval}
\Prob[\EZ(3)] \ge \Prob[\EA \cap \EB] = \Prob[\EA]\times \Prob[\EB|\EA] = \Omega\left(\frac{1}{\log T} \times \frac{1}{\log T} \right) = \Omega\left(\frac{1}{(\log T)^2}\right),
\end{equation}
as desired.
\end{proof}

Note that, by Lemma~\ref{lem_update}, $\EZ(3)$ implies $\EZ(4) \cap \EZ(5) \cap \dots \cap \EZ(T)$ with probability at least $1- 1/T$. It follows that $\Ex[\Regret(T)]$ is bounded as 
\begin{align}
\Ex[\Regret(T)] 
&\ge \Ex\left[\sum_{t=3}^T\regret(t)\Biggl|\EZ(3)\right] \Omega\left(\frac{1}{(\log T)^2}\right) \text{\ \ \ (by Eq.~\eqref{eq: ez eval})}\\
&\ge \left(1 - \frac{1}{T}\right)\Ex\left[\sum_{t=3}^T\regret(t)\Biggl|\bigcap_{t=3}^T \EZ(t)\right]
\Omega\left(\frac{1}{(\log T)^2}\right)\\
&\text{\ \ \ \ \ (by Lemma \ref{lem_update} and construction of $\EZ(3)$)}\\
&= \Theta(1) \times \Omega\left(\frac{T}{(\log T)^2}\right) \times 
\Omega\left(\frac{1}{(\log T)^2}\right)\\
&\text{\ \ \ \ \ (by Lemma \ref{lem_reglower_round}, $\EZ(t)$ implies $\Ex[\regret(t)] = \Omega(w_t^2) = \Omega(1/(\log T)^2)$)}\\
&=\Omega\left(\frac{T}{\polylog(T)}\right).
\end{align}

\end{proof} 

\subsection{Lemma \ref{lem_hatz_bound}}

We prove a lemma that is useful to prove Lemmas~\ref{lem_reglower_round} and \ref{lem_update}.

\begin{figure}
\vspace{6em}
	\centering
	\includegraphics[width = 0.8 \textwidth]{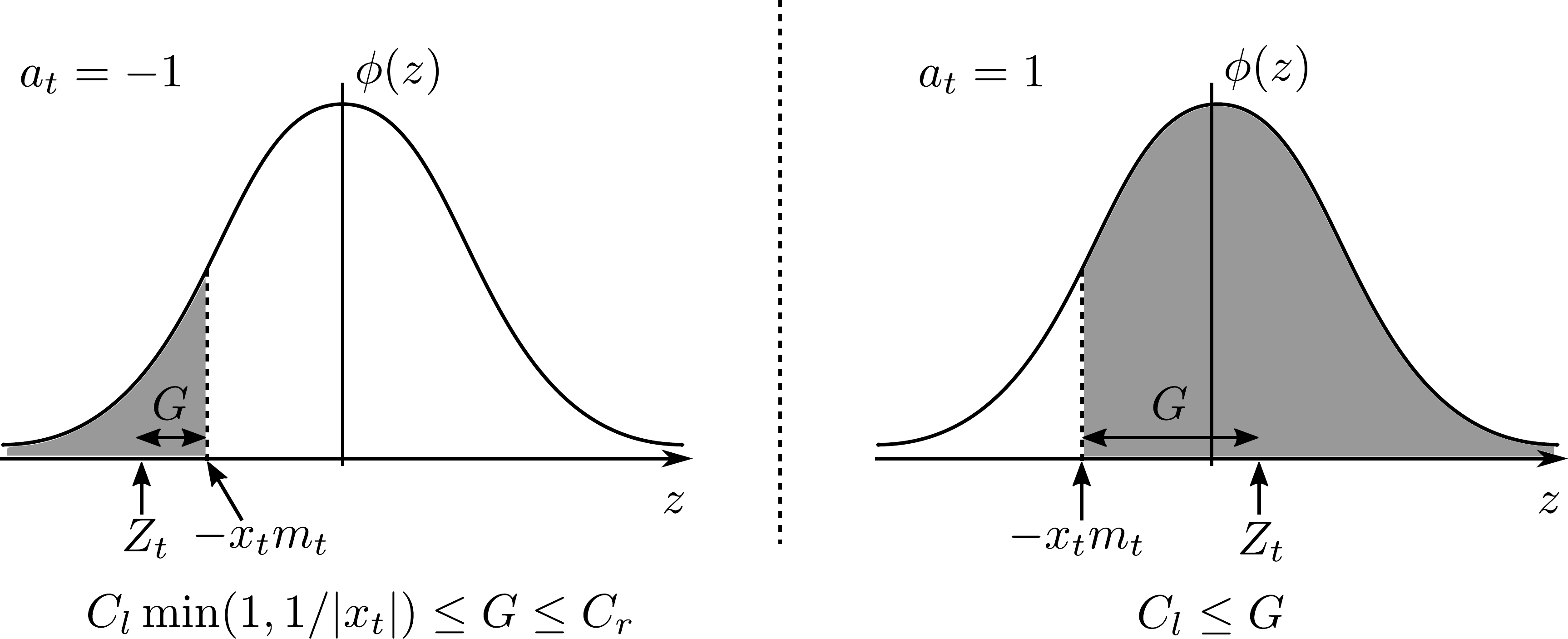}
	\caption{Illustration of Lemma \ref{lem_hatz_bound} when $x_t m_t > 0$. The lemma bounds $G := a_t( \Ehatzt + x_t m_t)$. The left figure corresponds to Eq.~\eqref{ineq_leftboth}, whereas the right figure corresponds to Eq.~\eqref{ineq_rightlower}. }
	\label{fig:ez_recom_gap}
\end{figure}

\begin{lem}[Gap between $Z_t$ and $-x_t m_t$: Binary Case]\label{lem_hatz_bound}
There exist universal constants $\Cezl, \Cezu > 0$ such that the following inequalities hold.
\begin{enumerate}
    \item If $\sgn(x_t m_t)a_t < 0$, then
    \begin{equation}\label{ineq_leftboth} 
    \Cezl\min\{1, 1/|x_t|\} < a_t( \Ehatzt + x_t m_t) < \Cezu.
    \end{equation}
    \item If $\sgn(x_t m_t)a_t > 0$, then
    \begin{equation} \label{ineq_rightlower}
    \Cezl < a_t(\Ehatzt+x_t m_t).
    \end{equation}
\end{enumerate}

\end{lem}
The term $\min\{1, 1/|x|\}$ in Eq.~\eqref{ineq_leftboth} is derived from the fact that $e^{-x^2/2}$ decays faster for a large $|x|$. It is analogous to the equation 
\begin{equation}
\frac{1}{(x+1)^2 - x^2} = \frac{1}{2x+1} \ge \frac{1}{3} \min\left\{1, \frac{1}{x}\right\}
\end{equation}
for $x>0$.

\begin{proof}[Proof of Lemma \ref{lem_hatz_bound}]

For ease of discussion, we assume $x_t m_t \ge 0$ (which aligns with Figure~\ref{fig:ez_recom}). The case of $x_t m_t < 0$ can be dealt with the same discussion.

Let $\phi(x) = \frac{1}{\sqrt{2 \pi}}e^{-x^2/2}$ be the pdf of the standard normal distribution and $\erf(x) = \frac{2}{\sqrt{\pi}} \int_0^x e^{-t^2}dt$ be the error function. Let $a_t = -1$ and $\Ctemp = \min\{1, {1}/{|x_t m_t|}\}$. Then,
\begin{align}
\Ehatzt + x_t m_t
&= \frac{\int_{-\infty}^{-x_t m_t} (z+x_t m_t) \phi(z) dz}{\int_{-\infty}^{-x_t m_t} \phi(z) dz} \\
&\le \frac{1}{\phi(-x_t m_t)} \int_{-\infty}^{-x_t m_t} (z+x_t m_t) \phi(z) dz \\
&\le \frac{1}{\phi(-x_t m_t)} \int_{-x_t m_t - 2\Ctemp}^{-x_t m_t - \Ctemp} (z+x_t m_t) \phi(z) dz \\
&\le \frac{-\Ctemp}{\phi(-x_t m_t)} \int_{-x_t m_t - 2\Ctemp}^{-x_t m_t - \Ctemp} \phi(z) dz \\
&\le \frac{-\Ctemp}{\phi(-x_t m_t)} \min_{z \in [-x_t m_t - 2\Ctemp, -x_t m_t - \Ctemp]} \phi(z) \\
&\le -\Ctemp \min_{z \in [-x_t m_t - 2\Ctemp, -x_t m_t - \Ctemp]} e^{-(3/2)} = -\Ctemp e^{-(3/2)}\\
&\text{\ \ \ \ \ (by $e^{-(x+a)^2/2}/e^{-x^2/2} = e^{-xa - a^2/2}$ and $|x_t m_t|\Ctemp \le 1$)}\\
& = -e^{-(3/2)} \min\{1, 1/|x_t m_t|\} \le -e^{-(3/2)} \min\{1, 1/|x_t|\},
\end{align}
which implies the first inequality\footnote{Note that $a_t = -1$ and the inequality here is flipped.} of Eq.~\eqref{ineq_leftboth}.

Moreover,
\begin{align}
\Ehatzt + x_t m_t
&= \Ex_{z \sim \trunkN_{-\infty, -x_t m_t}} [z] + x_t m_t \\
&= \frac{\int_{-\infty}^{-x_t m_t} (z+x_t m_t) \phi(z) dz}{\int_{-\infty}^{-x_t m_t} \phi(z) dz} \\
&\ge \frac{\int_{-\infty}^0 z \phi(z) dz}{\int_{-\infty}^{0} \phi(z) dz} \\
&\text{\ \ \ \ \ (by $\phi(x+c)/\phi(x) \le \phi(c)$ for any $x,c\le0$)}\\
&= - \frac{\int_{0}^\infty z \phi(z) dz}{\int_0^{\infty} \phi(z) dz} \\
&= -\sqrt{\frac{2}{\pi}},
\end{align}
which is a constant and implies the second inequality\footnote{Again, $a_t = -1$ and the inequality here is flipped.} of Eq.~\eqref{ineq_leftboth}.

If $a_t = 1$, then
\begin{align}
\Ehatzt + x_t m_t
&= \Ex_{z \in \trunkN(-x_t m_t, \infty)} [z] + x_t m_t\\
&\ge \frac{1}{2} \Ex_{z \in \trunkN(0, \infty)} [z] \\
&= \sqrt{\frac{1}{2\pi}},
\end{align}
which implies Eq.~\eqref{ineq_rightlower}.

\end{proof} 

\subsection{Proof of Lemma \ref{lem_reglower_round}}

\begin{proof}[Proof of Lemma \ref{lem_reglower_round}] 

\begin{figure}
\vspace{6em}
    \centering
    \includegraphics[width = 0.9 \textwidth]{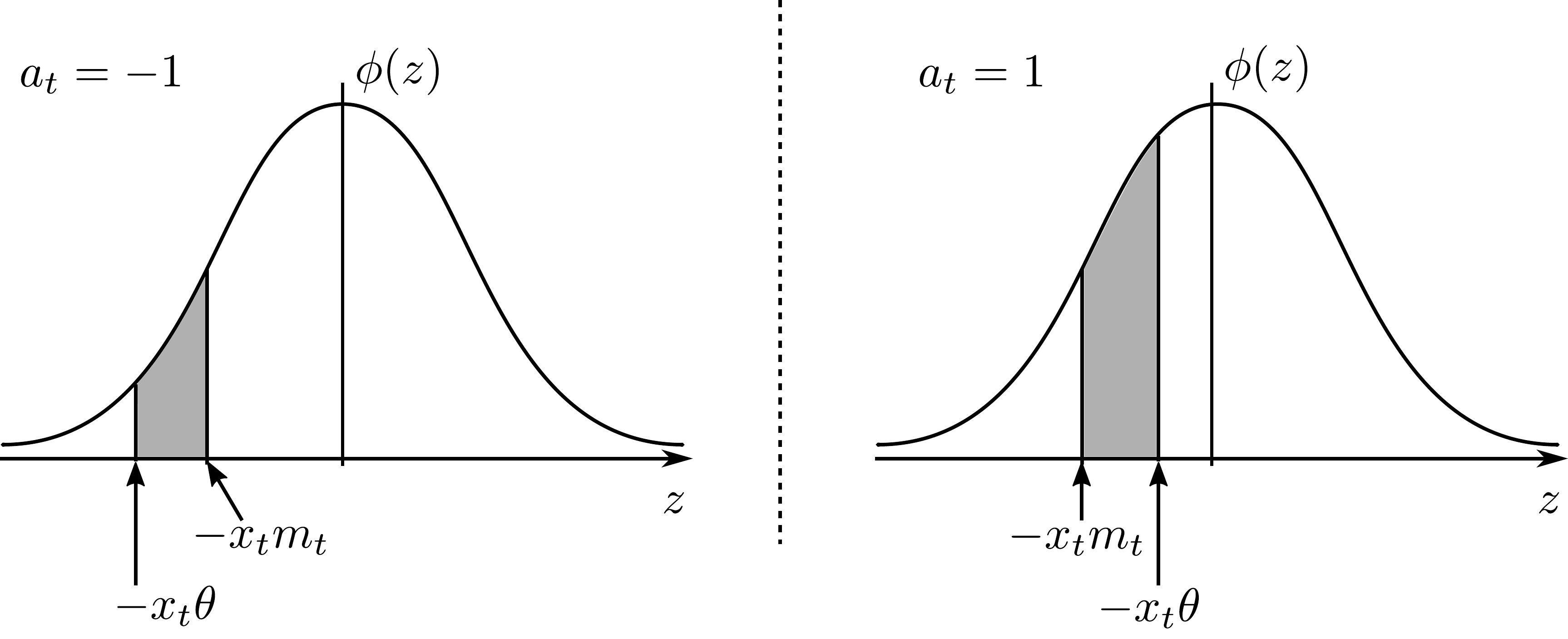}
    \caption{Recommendation $a_t$ is determined by the sign of $x_t m_t + z_t$, whereas the true superior arm is determined by $x_t \theta + z_t$. When $z_t \in [-x_t m_t, -x_t \theta]$, the recommender fails to recommend the superior arm.}
    \label{fig:ez_recom_lb}
\end{figure}

Without loss of generality, we assume $m_t \ge 0$. (Otherwise, by using the symmetry of the model, we may flip the sign of variables as $(l_t, u_t, \theta) = (-u_t, -l_t, -\theta)$ and apply the same analysis to obtain the same result.) For the ease of discussion, we assume $m_t - \theta > 0$. (In the case of $m_t - \theta < 0$, we can follow essentially the same discussion as the case of $m_t - \theta > 0$.)

For $t \in [T]$ and $l >0$, let 
\begin{equation}
\EC(t, l)
\coloneqq \{ - x_t m_t < z_t < -x_t \theta - l \} 
\cap 
\{ x_t (m_t - \theta) < \Cezl \}. 
\end{equation}

\begin{claim}\label{claim: reg lower bound A}
$\EC(t, l) \subseteq \{\reg(t) \ge l\}$.
\end{claim}

\begin{proof}
\begin{align}
\lefteqn{
\{ - x_t m_t < z_t < -x_t \theta - l \}
\cap 
\{ x_t (m_t - \theta) < \Cezl \}
}\\
&=
\{ 0 < x_t m_t + z_t \} \cap \{x_t \theta + z_t < -l\} 
\cap 
\{ x_t (m_t - \theta) < \Cezl \}
\\&\subseteq
\{ 0 < x_t m_t + z_t \} \cap \{x_t \theta + z_t < -l\} 
\cap 
\{ x_t \theta + \Ehatzt > 0\}\text{\ \ \ \ \ (by Eq.~\eqref{ineq_rightlower})} 
\\&=
\{ x_t \theta + z_t < -l\}
\cap \{ \ist = -1\} 
\cap \{ a_t = 1\} 
\cap \{ b_t = 1\}.
\end{align}
It follows from $b_t \neq \ist$ and $x_t \theta + z_t < -l$ that $\reg(t) > l$.
\end{proof}

\begin{claim}\label{claim: reg lower bound B}
$\Prob[\EC(t, l)]\ge \Theta (m_t - \theta - l)$.
\end{claim}

\begin{proof}
By using the fact that  $\Cezl$ is a universal constant and $1 \ge m_t > \theta \ge -1$, we have 
\begin{align}
\Prob[ x_t (m_t - \theta) < \Cezl ]
&\ge \Prob[ 2 x_t < \Cezl, x_t > 0] \\
&\ge \Prob[ \Cezl/2 < 2 x_t < \Cezl, x_t > 0] \\
&\ge \int_{\Cezl/4}^{\Cezl/2} \phi(x) dx  \\
&= \Theta(1).
\end{align}
Moreover, for any $\Cezl/4 \le x_t \le \Cezl/2$,
\begin{align}
\Prob[- x_t m_t < z_t < -x_t \theta - l]
&= \int_{- x_t m_t}^{-x_t \theta - l} \phi(z) dz\\
&= x_t (m_t - \theta - l) \min_{z \in \{- x_t m_t, -x_t \theta - l\}} \phi(z) \\
&= \Theta (m_t - \theta - l).
\end{align}
Therefore, $\Prob[\EC(t, l)] \ge \Theta(1) \times \Theta (m_t - \theta - l) = \Theta (m_t - \theta - l)$.
\end{proof}

Combining Claims~\ref{claim: reg lower bound A} and \ref{claim: reg lower bound B}, the regret is bounded as follows:
\begin{align}
\Ex[ \regret(t) ] 
&\ge \int_{0}^\infty \Prob[\regret(t) \ge l] dl \\
&\ge \int_{0}^\infty \Prob[\EC(t, l)] dl \text{\ \ \ (by Claim~\ref{claim: reg lower bound A})}\\
&\ge \int_{0}^{m_t - \theta} \Theta(m_t - \theta - l) dl \text{\ \ \ (by Claim~\ref{claim: reg lower bound B})}\\
&= \Omega((m_t - \theta)^2).
\end{align}
\end{proof} 

\subsection{Proof of Lemma \ref{lem_update}}

\begin{proof}[Proof of Lemma \ref{lem_update}]
Let
\begin{align}
\EU_1(t) &= \{b_t \sgn(x_t) < 0\} \cap \{b_t > -\Ehatzt/x_t\}, \\
\EU_2(t) &= \{b_t \sgn(x_t) > 0\} \cap \{a_t < -\Ehatzt/x_t\}, \\
\EU(t) &= \EU_1(t) \cup \EU_2(t).
\end{align}
By the update rule (Eq.~\eqref{ineq_update_left} and \eqref{ineq_update_right}), event $\EU(t)$ is equivalent to $(l_{t+1}, u_{t+1}) \ne (l_t, u_t)$.

Eq.~\eqref{ineq_leftboth} and \eqref{ineq_rightlower} in Lemma \ref{lem_hatz_bound} imply
\begin{equation}\label{lem_hatz_bound_reused}
|\Ehatzt - x_t m_t| \ge \Cezl \min\{1, 1/|x_t|\}.
\end{equation}
Accordingly,
\begin{align}
\EU(t) 
&= \EU_1(t) \cup \EU_2(t) \\
&\subseteq \{b_t > -\Ehatzt/x_t\} \cup \{a_t < -\Ehatzt/x_t\} \\
&\subseteq \{w_t/2 > \Cezl \min\{1/|x_t|, 1/|x_t|^2\}\} \\
&\text{\ \ \ (by $u_t - m_t = m_t - l_t = w_t/2$ and Eq.~\eqref{lem_hatz_bound_reused})}. 
\end{align}
For a sufficiently small $w_t$,\footnote{$w_t \le 2\Cezl$ is enough to assure $\{w_t/2 > \Cezl /|x_t|\} \subseteq \{w_t/2 > \Cezl /|x_t|^2\}$ because $\{w_t \le 2\Cezl, w_t/2 > \Cezl /|x_t|\}$ implies $|x_t| > 1$.}
\begin{align}
\Prob[\EU(t)] 
&\le \Prob\left[w_t > \frac{2\Cezl}{x_t^2}\right] \\
&= \Prob\left[x_t^2 > \frac{2\Cezl}{w_t}\right] \\
&= 2\Phi^c\left(\sqrt{\frac{2\Cezl}{w_t}}\right)\\
&\le 
\exp\left(-\frac{2\Cezl}{w_t}\right) \times 2\Phi^c\left(0\right)\\
&=\exp\left(-\frac{2\Cezl}{w_t}\right),
\end{align}
which completes the proof.
\end{proof} 

\subsection{Proof of Theorem \ref{thm_three}}
\begin{proof}[Proof of Theorem \ref{thm_three}]
Let 
\begin{equation}
\EE(t) = \left\{w_{t+1} \le \frac{5}{6} w_t\right\}.
\end{equation}
Lemmas \ref{lem_unknown_shrink} and \ref{lem_probunknown} imply that there exists a universal constant $\Cshrink > 0$ such that
\begin{equation}\label{ineq_eeprob}
\Prob[\EE(t)] \ge \Cshrink w_t.
\end{equation}
Lemma \ref{lem_regeps} states that  
\begin{equation}\label{ineq_regroundwt} 
\Ex[\regret(t)] \le \Cregttwo w_t^2.
\end{equation}

For $s = 1,2,\dots$, let
\begin{align}
\EP_s(t) &= \left\{\left(\frac{5}{6}\right)^s \le w_t \le \left(\frac{5}{6}\right)^{s-1}\right\},\\
\Reg_s(T) &= \sum_{t=1}^T \regret(t) \Ind\left\{\EP(t)\right\}.
\end{align}

Let $t_s$ be the first round in which $\EP_s(t)$ holds. Then, for each round $t=t_s+1,t_s+2,\dots$, we have the following:
\begin{enumerate}
    \item Eq.~\eqref{ineq_eeprob} implies that, with probability at least $\Cshrink (5/6)^{s - 1}$, $\EE(t)$ occurs. Furthermore, once $\EE(t)$ occurs, $\EP_s(t')$ never occurs again for round $t'>t$.
    \item Eq.~\eqref{ineq_regroundwt} implies that the expected regret per round is at most $\Cregttwo (5/6)^{2(s-1)}$.
\end{enumerate}
Accordingly, it follows that
\begin{align}
\Ex[\Reg_s(T)]
&\le \Cregttwo \left(\frac{5}{6}\right)^{2(s-1)}\sum_{u = 0}^\infty \left[1 - \Cshrink \left(\frac{5}{6}\right)^{s - 1}\right]^{u}\\
&= \frac{\Cregttwo}{\Cshrink}\left(\frac{5}{6}\right)^{s-1}.
\end{align}
The regret is bounded as
\begin{align}
\Ex[\Reg(T)]
&= \sum_{s=1}^\infty \Ex[\Reg_s(T)]\\
&\le \frac{\Cregttwo}{\Cshrink}\sum_{s=1}^\infty \left(\frac{5}{6}\right)^{s-1}\\
&= \frac{6\Cregttwo}{\Cshrink},
\end{align}
which is a constant.
\end{proof} 

\subsection{Proof of Lemma \ref{lem_unknown_shrink}}

\begin{proof}[Proof of Lemma \ref{lem_unknown_shrink}]
Let 
\begin{equation}
\EX(t) := \{x_t \ge 3\Ceps\}. 
\end{equation}

\begin{claim}\label{claim: Shrink A}
$\EX(t)$ and  $a_t=0$ implies $w_{t+1} \le (5/6)w_t$.
\end{claim}

\begin{proof}
Eq.~\eqref{ineq_update_left} and \eqref{ineq_update_right} imply that $l_{t+1} = \max\{l_t, -\Ehatzt/x_t\}$ or $u_{t+1} = \min\{u_t, -\Ehatzt/x_t\}$ always holds. 
By using this, we have
\begin{align}
&\{\EX(t), a_t=0\}\\
&\coloneqq \{x_t \ge 3\Ceps, a_t=0\}\\
&= \{x_t \ge 3\Ceps, |x_t m_t + \Ehatzt| \le \eps_t\}\\
&\subseteq \{|m_t + \Ehatzt/x_t| \le \eps_t/(3\Ceps)\}\\
&= \{|m_t + \Ehatzt/x_t| \le \eps_t/(3\Ceps)\} \cap
\{l_{t+1} = \max\{l_t, -\Ehatzt/x_t\} \cup u_{t+1} = \min\{u_t, -\Ehatzt/x_t\}\}\\
&\text{\ \ \ \ \ (by Eq.~\eqref{ineq_update_left} and \eqref{ineq_update_right})}\\
&\subseteq\{l_{t+1} \ge m_t - \eps_t/(3\Ceps)\,\cup\,u_{t+1} \le m_t + \eps_t/(3\Ceps) \} \\
&=\{l_{t+1} \ge m_t - (1/3)w_t\,\cup\,u_{t+1} \le m_t + (1/3)w_t \}.
\end{align}
Moreover, by $w_t/2 = u_t-m_t = m_t-l_t$, we have
\begin{equation}
\{l_{t+1} \ge m_t - (1/3)w_t\,\cup\,u_{t+1} \le m_t + (1/3)w_t \}
\subseteq \{w_{t+1} \le (5/6) w_t\}.
\end{equation}
\end{proof}

\begin{claim}\label{claim: Shrink B}
$\Prob[\EX(t),a_t=0] = \Theta(1)$.
\end{claim}

\begin{proof}
\begin{align}
\Prob[\EX(t),a_t=0] &= 
\int_{3\Ceps}^\infty \int_{-x_t m_t-\eps_t}^{-x_t m_t+\eps_t} \phi(z) \phi(x) dz dx\\
&\ge \eps_t \int_{3\Ceps}^\infty \phi(-x_t-\eps_t) \phi(x) dx \\
&\ge \eps_t \int_{3\Ceps}^\infty \phi(-x_t-1) \phi(x) dx \\
&\ge \frac{\eps_t}{2 \pi} \int_{3\Ceps}^\infty e^{-(x+1)^2} dx =: \Cgeomuptmp \ \eps_t.
\end{align}
\end{proof}

Combining Claims~\ref{claim: Shrink A} and \ref{claim: Shrink B}, we have
\begin{align}
&\Prob\left[\left.w_{t+1} \le \frac{5}{6} w_t\right|a_t=0\right]\\
&=\frac{\Prob\left[w_{t+1} \le \dfrac{5}{6} w_t,a_t=0\right]}{\Prob[a_t=0]}\\
&\ge\frac{\Prob\left[w_{t+1} \le \dfrac{5}{6} w_t,a_t=0\right]}{\CfenceU \eps_t} \text{\ \ \ \ (by Lemma \ref{lem_probunknown})}\\
&\ge \frac{\Prob[\EX(t),a_t=0]}{\CfenceU \eps_t} \text{\ \ \ (by Claim~\ref{claim: Shrink A})}\\
&\ge \frac{\Cgeomuptmp}{\CfenceU \eps_t} \text{\ \ \ (by Claim~\ref{claim: Shrink B})}\\
&\eqqcolon \Cgeomup.
\end{align}

\end{proof} 

\subsection{Proof of Lemma \ref{lem_probunknown}}

\begin{proof}[Proof of Lemma \ref{lem_probunknown}]
We have
\begin{align}
\Prob[a_t = 0] 
&= \Prob\left[ |x_t m_t + z_t| \le \eps_t \right] \\
&= \Prob\left[
\int_{-\eps_t/(1+m_t^2)}^{\eps_t/(1+m_t^2)} \phi(x) dx
\right]
\text{\ \ \ (by $x_t m_t + z_t \sim \Normal(0, 1+m_t^2)$ given $m_t$)} \\
&= \Theta(\eps_t) 
\text{\ \ \ (by $1 \le (1+m_t^2) \le 2$ and $\phi(x) \le 1$),}
\end{align}
which completes the proof.
\end{proof} 

\subsection{Proof of Lemma \ref{lem_regeps}}

We first introduce the following lemmas.
\begin{lem}[Gap between $Z_t$ and $-x_t m_t$: Ternary Case]\label{lem_hatz_bound_ternary}
There exist universal constants $\Cezl > 0$ such that the following inequalities hold.
\begin{enumerate}
    \item If $\sgn(x_t m_t)a_t < 0$, then
    \begin{equation}\label{ineq_leftlower_ternary} 
    \Cezl\min\{1, 1/|x_t|\} < a_t( \Ehatzt + x_t m_t).
    \end{equation}
    \item If $\sgn(x_t m_t)a_t > 0$, then 
    \begin{equation} \label{ineq_rightlower_ternary}
    \Cezl < a_t(\Ehatzt+x_t m_t).
    \end{equation}
\end{enumerate}
\end{lem}
Lemma \ref{lem_hatz_bound_ternary} is a version of Lemma \ref{lem_hatz_bound} for the ternary recommendation.
We omit the proof of Lemma \ref{lem_hatz_bound_ternary} because it follows the same steps as Lemma \ref{lem_hatz_bound}.

\begin{lem}[Expected Regret from Choosing the Inferior Arm]\label{lem_three_mistake}
The following inequality holds:
\begin{equation}
\Ex\left[\regret(t)\Ind\left\{\ist \ne b_t, a_t\ne 0\right\}\right] = O(w_t^2).
\end{equation}
\end{lem}%

\begin{proof}[Proof of Lemma \ref{lem_three_mistake}]
We have
\begin{align}
\{\ist \ne b_t, a_t \ne 0\} 
\subseteq \{\ist \ne a_t, a_t \ne 0\} \cup \{a_t \ne b_t, a_t \ne 0\},\label{ineq_wronganswer_three}
\end{align}
and we bound each of the terms on the right-hand side.

\begin{claim}\label{claim: three mistake A}
$\Ex[\regret(t)\Ind\{\ist \ne a_t, a_t\ne 0\}] = O(w_t^2)$.
\end{claim}

\begin{proof}
\begin{align}
\{\ist \ne a_t, a_t \ne 0\} 
&\subseteq \{\ist \ne a_t\} \\
&= \{\sgn(x_t \theta + z_t) \ne \sgn(x_t m_t + z_t)\}\\
&=  \{z_t \in [\min\{-x_t \theta, -x_t m_t\}, \max\{-x_t \theta, -x_t m_t\}]\},
\end{align}
and thus, conditioning on $x_t$, we have
\begin{align}
\lefteqn{
\Prob[ z_t \in [\min\{-x_t \theta, -x_t m_t\}, \max\{-x_t \theta, -x_t m_t\}]\,| x_t ] 
}\\
&\le \int_{\min\{-x_t \theta, -x_t m_t\}}^{\max\{-x_t \theta, -x_t m_t\}} \phi(z) dz\\
&\le \int_{\min\{-x_t \theta, -x_t m_t\}}^{\max\{-x_t \theta, -x_t m_t\}} dz = |x_t(\theta - m_t)|,\label{ineq_givenxt}
\end{align}
where we have used the fact that $\phi(z) \le 1 $. The event $\ist \ne a_t$ implies $\reg(t) \le x_t w_t$, and
marginalizing Eq.~\eqref{ineq_givenxt} over $x_t$, we have
\begin{align}
\Ex\left[\regret(t)\Ind\left\{\ist \ne a_t, a_t\ne 0\right\}\right]
&\le \int_{-\infty}^\infty \phi(x) |x^2 w_t (\theta - m_t)| dx\\
&\le \int_{-\infty}^\infty \phi(x) x^2 w_t^2 dx\\
&= w_t^2 \int_{-\infty}^\infty \phi(x) x^2 dx \\
&= O\left(w_t^2\right),
\end{align}
as desired.
\end{proof}

\begin{claim}\label{claim: three mistake B}
$\Prob[a_t \ne b_t, a_t \ne 0] = O(w_t^2)$.
\end{claim}

\begin{proof}
\begin{align}
&\{a_t \ne b_t, a_t \ne 0\}\\
&= \{\sgn(x_t m_t + z_t) \ne \sgn(x_t \theta + \Ehatzt), a_t \ne 0\}\\
&\subseteq \{x_t m_t + z_t>0, x_t \theta + \Ehatzt<0\} 
\cup \{x_t m_t + z_t<0, x_t \theta + \Ehatzt > 0\} \\
&\subseteq \{x_t \theta - x_t m_t + \Cezl\min\{1,1/|x_t|\} < 0 \} 
\cup \{x_t \theta - x_t m_t - \Cezl\min\{1,1/|x_t|\}> 0\} \\
&\text{\ \ \ \ \ (by Eq.~\eqref{ineq_leftlower_ternary} and Eq.~\eqref{ineq_rightlower_ternary})}\\
&= \{|x_t \theta - x_t m_t| \ge \Cezl\min\{1,1/|x_t|\} \},
\end{align}
and thus
\begin{align}
\Prob[a_t \ne b_t, a_t \ne 0]
&\le \Prob[|x_t \theta - x_t m_t| \ge \Cezl\min\{1,1/|x_t|\}  ]\\
&= \Prob\left[|x_t|^2 \ge \frac{\Cezl}{|\theta - m_t|} \right]\\
&\le \Prob\left[|x_t|^2 \ge \frac{\Cezl}{w_t} \right]\\
&= 2\Phi^c\left(\sqrt{\frac{\Cezl}{w_t}}\right)\\
&\le e^{-\frac{w_t}{2\Cezl}}\\
&= O(w_t^2).
\text{\ \ \ (An exponential decays faster than any polynomial)}
\end{align}
\end{proof}

\noindent(Proof of Lemma~\ref{lem_three_mistake}, continued.) Combining Claims~\ref{claim: three mistake A} and \ref{claim: three mistake B}, we have
\begin{align}
\Ex\left[\regret(t)\Ind\left\{\ist \ne b_t, a_t\ne 0\right\}\right]
&\le 
\Ex\left[\regret(t)\Ind\left\{\ist \ne a_t, a_t\ne 0\right\}\right]
+
\Ex\left[\regret(t)\Ind\left\{a_t \ne b_t, a_t\ne 0\right\}\right]
\\
&=O(w_t^2). 
\end{align}
\end{proof} 

\begin{proof}[Proof of Lemma \ref{lem_regeps}]
\begin{align}
\Ex[\regret(t)]
&\le \Ex\left[\Ind\{a_t = 0\}\regret(t)\right]+ \Ex\left[\Ind\left\{b_t \ne \ist, a_t \ne 0\right\}\regret(t)\right]\\ 
&\le \Ex\left[\Ind\{a_t = 0\}|x_t \theta + z_t|\right]+ \Ex\left[\Ind\left\{b_t \ne \ist, a_t \ne 0\right\}\regret(t)\right]\\
&\le \Prob[a_t = 0](\eps_t+w_t) + \Ex\left[\Ind\left\{b_t \ne \ist, a_t \ne 0\right\}\regret(t)\right]\\ 
&\text{\ \ \ \ \ (by $a_t = 0$ implies $|x_t m_t + z_t| \le \eps_t$ and $|x_t \theta + z_t| - |x_t m_t + z_t| \le |x_t w_t|$)}
\\
&\le O((\eps_t+w_t)^2) + \Ex\left[\Ind\left\{b_t \ne \ist, a_t \ne 0\right\}\regret(t)\right]
\text{\ \ \ (by Lemma \ref{lem_probunknown})}\\
&\le O((\eps_t+w_t)^2) + O(w_t^2)
\text{\ \ \ (by Lemma \ref{lem_three_mistake})}\\
& = O((\max\{\eps_t,w_t\})^2).
\end{align}

\end{proof}

\subsection{Proof of Theorem~\ref{thm_binary_shortopt}}

\begin{proof}[Proof of Theorem \ref{thm_binary_shortopt}]

For $c \ge 1$, it follows from \citep{feller-vol-1} that
\begin{equation}
    \Prob[|x_t| \ge c] \le e^{-\frac{c^2}{2}},
\end{equation}
and thus
\begin{equation}
    \Prob\left[|x_t| \ge \sqrt{2\log(1/\delta)}\right] \le \delta.
\end{equation}
Therefore, if $w_t \le \delta^2/(4 \sqrt{\pi \log(1/\delta)})$, then with probability $1-\delta$, we have
\begin{equation}\label{ineq_deltawt}
\frac{1}{2\sqrt{2 \pi }}e^{-x_t^2} \ge  |x_t| w_t
\end{equation}
for all $\theta \in [-1, 1]$.

Consider an arbitrary threshold policy $\rho_t \neq \rho_t^{\mathrm{st}}$.
When the user always follows the recommendation (i.e., $\rho_t$ belongs to Case~4), from the perspective of the recommender, the user's expected payoff is
\begin{equation}
    \Ex_{\tiltheta \sim \Unif[l_t, u_t], z_t \sim \Normal}\left[\Ind\left\{a_t = 1\right\}\left(x_t \tiltheta + z_t\right)\right] = x_t m_t (1 - \Phi(\rho_t)) + \phi(\rho_t),
\end{equation}
and this formula takes maximum at $\rho_t^{\mathrm{st}} \coloneqq -x_t m_t$. Accordingly, whenever the straightforward policy is suboptimal, (i.e., there exists $\rho_t$ such that $V(\rho_t) > V(\rho_t^{\mathrm{st}})$), then there exists $\theta^{\mathrm{d}} \in [l_t, u_t]$ such that the user deviates from the recommendation at $\theta^{\mathrm{d}}$ given $\rho_t$. 
Note that the only possible strategies for the user for a fixed $(x_t, \theta)$ is (i) to follow the recommendation (i.e., $b_t = a_t$), (ii) to deviate to $-1$ (i.e., $b_t = -1$ regardless of $a_t$), or (iii) to deviate to $1$ (i.e., $b_t = 1$ regardless of $a_t$).

We examine the latter two cases of deviations.

\noindent\textbf{Case A ($b_t = -1$ while $a_t = 1$):} 
In this case, at $\theta^{\mathrm{d}}$, the user chooses arm $-1$ even when $a_t = 1$. Since the user takes $b_t = - 1$ deterministically, his expected payoff (computed before receiving $a_t$) is fixed to $0$. Accordingly, the user adopts this strategy at $\theta^{\mathrm{d}}$ if and only if
\begin{align}
    \Ex_{z_t \sim \Normal}\left[\Ind\left\{z_t \ge \rho_t \right\}(x_t \theta^{\mathrm{d}} + z_t) \right] & < 0,\\
    \left(\Ex_{z_t \sim \Normal}\left[x_t \theta^\mathrm{d} + z_t \right] =\right) x_t \theta^\mathrm{d}  & < 0.
\end{align}
Using $|\theta - \theta^{\mathrm{d}}| \le w_t$, we have
\begin{align}
    \Ex_{z_t \sim \Normal}\left[\Ind\left\{z_t \ge \rho_t\right\}(x_t \theta + z_t)\right] & \le |x_t| w_t \\
    x_t \theta & \le |x_t| w_t,
\end{align}
for all $\theta \in [l_t, u_t]$. Accordingly, this user's expected payoff under $\rho_t$ is bounded as
\begin{align}
&\Ex_{\tiltheta \sim \Unif[l_t, u_t], z_t \sim \Normal}\left[\Ind\left\{b_t = 1\right\}\left(x_t \tiltheta + z_t \right)\right]\\
&\le 
\Ex_{\tiltheta \sim \Unif[l_t, u_t]}\left[
\max\left\{
0,
\Ex_{z_t \sim \Normal}\left[\Ind\left\{z_t \ge \rho_t\right\}(x_t \tiltheta + z_t)\right], x_t \tiltheta
\right\}
\right]
\\
&\le |x_t| w_t.
\end{align}

Meanwhile, the user's expected payoff from the straightforward policy given $x_t$ is bounded as follows
\begin{align}
\label{ineq_zerodeviate}
&\Ex_{\tiltheta \sim \Unif[l_t, u_t], z_t \sim \Normal}\left[\Ind\left\{z_t \ge -x_t m_t\right\}\left(x_t \tiltheta + z_t\right)\right]\\ 
&=
\int_{-x_t m_t}^{\infty} (x_t m_t + z_t) \phi(z_t) dz_t 
\\
&= \int_{0}^{\infty} z' \phi(z'-x_t m_t) dz'
\\
& =  \int_{0}^{\infty} z' \frac{1}{\sqrt{2\pi}} e^{-(z' - x_t m_t)^2/2} dz'\\
&\ge \frac{1}{\sqrt{2\pi}} e^{-(x_t m_t)^2} \int_{0}^{\infty} z' e^{-(z')^2} dz'
\text{\ \ \ \ (by $x^2+y^2 \ge (x+y)^2/2$)}
\\
&= \frac{1}{2\sqrt{2\pi}} e^{-(x_t m_t)^2}
\\
&\ge |x_t| w_t \text{\ \ \ (by Eq.~\eqref{ineq_deltawt})},
\end{align}
implying that $\rho_t$ is suboptimal.

\noindent\textbf{Case~B ($b_t = 1$ while $a_t = -1$):} 
Most of the argument is parallel to Case~A. 
In this case, the user chooses arm $1$ regardless of $a_t \in \{-1, 1\}$ at $\theta^{\mathrm{d}}$. Such a choice is optimal for the user if and only if
\begin{align}
    \Ex_{z_t \sim \Normal}\left[\Ind\left\{z_t \ge \rho_t^{\mathrm{myopic}}\right\}(x_t \theta^{\mathrm{d}} + z_t)\right]
    &< x_t \theta^{\mathrm{d}}
    \\
    0
    &< x_t \theta^{\mathrm{d}}.
\end{align}
Using $|\theta - \theta^{\mathrm{d}}| \le w_t$, we have
\begin{align}
\Ex_{z_t \sim \Normal}\left[\Ind\left\{z_t \ge \rho_t^{\mathrm{myopic}}\right\}(x_t \theta + z_t)\right]
&\le x_t \theta + |x_t| w_t\\
0 &\le x_t \theta + |x_t| w_t,
\end{align}
for all $\theta \in [l_t, u_t]$, which implies that the expected payoff is at most
\begin{equation}
    \Ex_{\tiltheta \sim \Unif[l_t, u_t], z_t \sim \Normal}\left[\Ind\left\{b_t = 1\right\}\left(x_t \tiltheta + z_t \right)\right] \le x_t \theta + |x_t| w_t.
\end{equation}

Meanwhile, the expected payoff when the user follows the straightforward recommendation is bounded as
\begin{align}
&\Ex_{\tiltheta \sim \Unif[l_t, u_t], z_t \sim \Normal}\left[\Ind\left\{z_t \ge -x_t m_t\right\}\left(x \tiltheta + z_t\right)\right]
\\
&= x_t m_t
-\int_{-\infty}^{-x_t m_t} (x_t m_t + z_t) \phi(z_t) dz_t \\
&= x_t m_t + 
\int_{0}^{\infty} z' \phi(-x_t m_t - z') dz'\\
&\ge x_t m_t + \frac{1}{2\sqrt{2\pi}} e^{-(x_t m_t)^2}\\
&\ge x_t m_t + |x_t| w_t, \text{\ \ \ (by Eq.~\eqref{ineq_deltawt})}
\end{align}
implying that $\rho_t$ is suboptimal.

In summary, (i) $x_t \le \sqrt{2\log(1/\delta)}$ occurs with probability at least $1 - \delta$, and (ii) when it occurs, if $w_t \le \delta^2/(4 \sqrt{\pi \log(1/\delta)})$, then $\rho_t^{\mathrm{st}} = - x_t m_t$ is optimal.
\end{proof}

\subsection{Proof of Theorem~\ref{thm: EvE regret bound}}

We first prove the following lemma.

\begin{lem}[Geometric Update by EvE]\label{lem: Eve Geometric Update}
Under the exploration phase of the EvE policy, we have
\begin{equation}
    \Prob\left[\left.w_{t+1} < \frac{1}{2} w_t \right| l_t, u_t \right] > C_{\mathrm{EvE update}}.
\end{equation}
for all $l_t, u_t$.
\end{lem}

\begin{proof}[Proof of Lemma \ref{lem: Eve Geometric Update}]
We prove the case of $m_t > 0$. The proof for the other case is similar. When $z_t < c(x_t m_t)$, the recommender sends $a_t = - 1$, and user $t$'s expected payoff from arm $1$ is $x_t (\theta - m_t)$. Accordingly, for any case, the confidence interval is halved.

We evaluate the probability that $z_t < c(x_t m_t)$ occurs. Since $c(x_t m_t)$ is defined to satisfy $\Ex[z'_t | z'_t < c(x_t m_t)] = - x_t m_t$, we have $c(x_t m_t) > - x_t m_t$. Since $x_t \sim \Normal$, with probability $\Phi(1) - \Phi(0)$, $x_t \in (0, 1)$. For such $x_t$, $- x_t m_t \in (- m_t, 0)$, and therefore, $- x_t m_t < - m_t \le - 1$. Accordingly, for any $z_t < - 1$, we have $z_t < c(x_t m_t)$. Accordingly, with probability at least $(\Phi(1) - \Phi(0))\Phi(-1) = \Theta(1)$, $w_{t+1} < w_t/2$ occurs.
\end{proof}

\begin{proof}[Proof of Theorem~\ref{thm: EvE regret bound}]
By Lemma~\ref{lem: Eve Geometric Update}, the confidence interval shrinks geometrically with a constant probability, and it takes $O(\log T)$ rounds in expectation to have $w_t < 1/\sqrt{T}$ to terminate the exploration phase. During the exploration phase, the per-round regret is $O(1)$ in expectation, and therefore, $O(\log T)$ regret is incurred. During the exploitation phase, the per-round regret is $O(w_t^2) = O(1/T)$ (implied by Lemma~\ref{lem_reglower_round}), and therefore, the total regret is $O(1/T \times T) = O(1)$. Accordingly, $O(\log T)$ regret is incurred in total.
\end{proof}

\subsection{Proof of Theorem~\ref{lem_deviation}}

\begin{proof}[Proof of Theorem~\ref{lem_deviation}]
For ease of discussion, we assume $x_t > 0$. The case of $x_t < 0$ can be proved in a similar manner.

\begin{case}
$a_t = b_t = 1$.
\end{case}
We have
\begin{align}
w_{t+1} 
&= u_{t+1} - l_{t+1}\\
&= u_t - \max\{l_t, (-\Ehatzt/x_t)\}
\text{\ \ \ (by Eq.~\eqref{ineq_update_left} and \eqref{ineq_update_right})}\\
&> u_t - m_t \text{\ \ \ (by $m_t > l_t$ and $a_t = 1$ if and only if $m_t > - z_t/x_t$)}\\
&= \frac{1}{2}w_t. \label{ineq_dev_rr}
\end{align}

\begin{case}
$a_t = b_t = -1$
\end{case}
We have
\begin{align}
w_{t+1} 
&= u_{t+1} - l_{t+1}\\
&= \min\{u_t, (-\Ehatzt/x_t)\} - l_t
\text{\ \ \ (by Eq.~\eqref{ineq_update_left} and \eqref{ineq_update_right})}\\
&< m_t - l_t \text{\ \ \ \text{\ \ \ (by $m_t < u_t$ and $a_t = -1$ if and only if $m_t < - z_t/x_t$)}}\\
&= \frac{1}{2} w_t. \label{ineq_dev_ll}
\end{align}

Eq.~\eqref{ineq_wtsize_follow} follows from Eq.~\eqref{ineq_dev_rr} and \eqref{ineq_dev_ll}.

\begin{case}
$a_t = -1$, $b_t = 1$
\end{case}
We have
\begin{align}
w_{t+1} 
&= u_{t+1} - l_{t+1}\\
&\le u_t - (-\Ehatzt/x_t)
\text{\ \ \ (by Eq.~\eqref{ineq_update_left} and \eqref{ineq_update_right})}\\
&< u_t - m_t \text{\ \ \ (by $a_t=-1$)}\\
&= \frac{1}{2} w_t. \label{ineq_dev_lr}
\end{align}

\begin{case}
$a_t = 1$, $b_t = -1$
\end{case}
We have
\begin{align}
w_{t+1} 
&= u_{t+1} - l_{t+1}\\
&\le (-\Ehatzt/x_t) - l_t
\text{\ \ \ (by Eq.~\eqref{ineq_update_left} and \eqref{ineq_update_right})}\\
&< m_t - l_t \text{\ \ \ (by $a_t=1$)}\\
&= \frac{1}{2} w_t. \label{ineq_dev_rl}
\end{align}

Eq.~\eqref{ineq_wtsize_oppose} follows from Eq.~\eqref{ineq_dev_lr} and \eqref{ineq_dev_rl}.

\end{proof}

\end{document}